\documentclass[sigconf]{acmart}

\AtBeginDocument{%
  \providecommand\BibTeX{{%
    \normalfont B\kern-0.5em{\scshape i\kern-0.25em b}\kern-0.8em\TeX}}}

\setcopyright{none}
\renewcommand\footnotetextcopyrightpermission[1]{} % removes footnote with conference information in first column
\usepackage{float}
\usepackage{balance}
\usepackage{bm}
\usepackage{amsmath,amsthm,amsmath}
\usepackage{algorithm, pseudocode}
\usepackage{mathrsfs}

\def\K{\ensuremath{\mathbb{K}}}
\def\Kbar {\ensuremath{\overline{\mathbb{K}}}}
\DeclareBoldMathCommand{\c}{c}
\DeclareBoldMathCommand{\a}{a}
\DeclareBoldMathCommand{\f}{f}
\DeclareBoldMathCommand{\g}{g}
\DeclareBoldMathCommand{\h}{h}
\DeclareBoldMathCommand{\x}{x}
\DeclareBoldMathCommand{\z}{z}
\DeclareBoldMathCommand{\v}{v}
\DeclareBoldMathCommand{\0}{0}
\DeclareBoldMathCommand{\u}{u}
\DeclareBoldMathCommand{\e}{e}
\DeclareBoldMathCommand{\p}{p}
\DeclareBoldMathCommand{\q}{q}
\DeclareBoldMathCommand{\s}{s}

\def\b_eta{\mbox{\boldmath$\eta$}}
\def\softO{\ensuremath{{O}{\,\tilde{ }\,}}}
\DeclareBoldMathCommand{\balpha}{\alpha}
\def\jac{\ensuremath{{\rm Jac}}}
\def\diag{\ensuremath{\mathrm{diag}}}
\def\jac{\ensuremath{{\rm Jac}}}
\def\rank{\ensuremath{{\rm rank}}}

\def\G {\ensuremath{\mathbf{G}}}
\def\J {\ensuremath{\mathbf{J}}}
\def\X {\ensuremath{\mathbf{X}}}
\def\Z {\ensuremath{\mathbf{Z}}}

\def\scrR{\ensuremath{\mathscr{R}}}
\def\scrL{\ensuremath{\mathscr{L}}}

\def\scrQ{\ensuremath{\mathscr{Q}}}

\newcommand{\wbigcup}{\mathop{{\bigcup}}\displaylimits}

\DeclareBoldMathCommand{\bvartheta}{\vartheta}
\DeclareBoldMathCommand{\bkappa}{\kappa}
\DeclareBoldMathCommand{\btheta}{\theta}

\def\partition{\ensuremath{(m_1^{k_1} \, \dots \,
m_r^{k_r})}}

\ccsdesc[500]{Computing methodologies~Algebraic algorithms}
\ccsdesc[500]{Theory of computation~Design and analysis of algorithms}

\begin{document}
\nocite{*}
%\title{Computing critical points for  algebraic systems defined by polynomials invariant under the action of the hyperoctahedral groups}
\title{Computing critical points for  algebraic systems defined by hyperoctahedral invariant polynomials}


\author{Thi Xuan Vu}
\affiliation{%
  \institution{Department of Mathematics and Statistics \\ UiT The Arctic University of Norway }
  \city{9037 Troms\o{}}
  \country{Norway}}
\email{thi.x.vu@uit.no}

\renewcommand{\shortauthors}{Thi Xuan Vu}

\begin{abstract}
Let $\K$ be a field of characteristic zero and $\K[x_1, \dots, x_n]$ the corresponding  multivariate
polynomial ring.
Given  a sequence of $s$ polynomials
$\f = (f_1, \dots, f_s)$ and a polynomial $\phi$, all in  $\K[x_1, \dots, 
x_n]$ with $s<n$, we consider the problem of  computing the set
$W(\phi, \f)$ of points at which $\f$ vanishes and the Jacobian matrix
of $\f, \phi$ with respect to $x_1, \dots, x_n$ does not have full
rank.  This problem plays an essential role in many application areas.   

 In this paper we focus on case where the polynomials are all invariant under the action of the signed symmetric group $B_n$.  
We introduce a notion called {\em hyperoctahedral representation} to
describe  $B_n$-invariant sets. We study the invariance
properties of the input polynomials to split $W(\phi, \f)$  according to the
orbits of $B_n$ and then design  an algorithm whose output is a
{hyperoctahedral representation} of $W(\phi, \f)$. The runtime of our
algorithm is polynomial in the total number of points described by the
output. 
\end{abstract}

\keywords{Invariant polynomials, 
  hyperoctahedral groups, critical points, complexity analysis}
\maketitle

\section{Introduction}
 Let $\K$ be a field of characteristic zero with
 algebraic closure $\Kbar$ and let ${C}$ be a vector space of finite dimension
 over $\K$. An invertible linear transformation $\Theta : C
 \rightarrow C$ is a reflection if all but one of its eigenvalues are
 $1$ and the non-1 eigenvalue is $-1$. We denote by $\K[{C}]$ the ring
 of polynomial functions of ${C}$ with coefficients in $\K$.  
 
A finite group $G$ of invertible linear transformations of ${C}$ is a
reflection group if it is generated by  reflections. 
 We define the set $$\K[C]^G := \{ f \in \K[C] \, : \, \sigma(f) = f \, {\rm
   for \ all} \, \sigma \in G\}$$ to be the polynomial functions
 invariant under  the action of $G$.

%Let $G$ be a finite group. 
The Chevalley-Shephard-Todd Theorem (see
e.g.~[9, Chapter~6]) states that any finite group $G$ is generated by
reflections if and only if  $\K[C]^G$ is a  polynomial ring in $\dim
{C}$ homogeneous generators. In other words, picking a dual basis
$(x_1, \dots, x_n)$ for $C^*$, where $n= \dim(C)$, there exists a sequence of $n$
homogeneous polynomials $(\theta_1, \dots, \theta_n)$ in $\K[x_1, \dots, x_n]$
such that for any $f \in \K[x_1, \dots, x_n]^G$, there exists a
polynomial $\bar{f} \in \K[y_1, \dots, y_n] $ such that  
 \[
 \bar{f}(\theta_1, \dots, \theta_n) = f(x_1, \dots, x_n).
 \]

 %%%%%%%%%%%%%

\paragraph{Problem and motivations}

 An important type of reflection groups is the hyperoctahedral
 groups. The hyperoctahedral
 groups have applications in quantum chromodynamics, for example, in the classification of states of quarks. These groups also appear in many other fields of physics such as in the symmetry groups of non-rigid molecules, non-rigid water clusters (see e.g. [3]). 
 
 Given a positive integer $n$. The hyperoctahedral group $B_n$ is the group of all signed
 permutations. This group contains all bijections $\tau$ of the set
 $\{\pm 1, \dots, \pm n\}$ to itself such that $\tau(-i) =
 -\tau(i)$ for all $i \in \{\pm 1, \dots, \pm n\}$. The
 hyperoctahedral group $B_n$ acts on $\K[x_1, \dots, x_n]$ by
 permuting the $x_i$'s and changing the signs of an arbitrary number of
 variables. The invariant ring $\K[x_1, \dots, x_n]^{B_n}$ is then
 generated by $(\eta_1(x_1^2, \dots, x_n^2), \dots, \eta_n(x_1^2, \dots, $
 $x_n^2))$, where $\eta_k(\cdot)$ denotes the $k$-th elementary symmetric
 function (see Subsection~\ref{subsec:Hyper_groups} for more
 details). 
 
We assume that we have a sequence of $s$ polynomials $\f = (f_1, \dots, f_s)$ and a
polynomial $\phi$ in $\K[x_1,\dots, x_n]$, with $s<n$, which are all invariant under
the action of the hyperoctahedral group. In this paper, we are interested in describing the set  
 \[
 W(\phi, \f) = \left\{\x \in \Kbar{}^n : \f(\x) = 0 {\rm
     ~and~} \rank(\jac_\X(\f, \phi)(\x)) < s+1  \right\},  
 \] where $\jac_\X(\f, \phi)$ is the Jacobian matrix of $(\f, \phi)$
 with respect to variables $\X = (x_1, \dots, x_n)$. We assume further
 that  this set is finite. This allows us to compare the differences
 between the number of points computed by using our algorithm and the
 whole finite set $W(\phi, \f)$. {This finiteness assumption is not very
 restrictive: instead of using this
 assumption, one can compute the isolated points set which also
 contains finitely many points}.

The set $W(\phi, \f)$ is an algebraic set which is defined by $\f$
 and all the maximal minors of $\jac_{\X}(\f, \phi)$. Note that
 although $\f$ and $\phi$ are $B_n$-invariant, the equations defining
 $W(\phi, \f)$ are not necessarily $B_n$-invariant. However, we will see in
 Section~\ref{sec:main_alg} that $W(\phi, \f)$ is invariant under the
 action of $B_n$, that is, 
 %A finite set $S \subset \Kbar{}^{n}$ is said to be $B_n$-invariant 
 for all $\tau \in B_n$ and
 $\s \in W(\phi, \f)$
 then $\tau(\s) \in W(\phi, \f)$. 
 \begin{example} \label{ex:import}
 Let $n=3, s=1$, and $\phi = x_1^2x_2^2x_3^2 - 3x_1^2 - 3x_2^2 -
3x_3^2$ and $f = x_1^4+x_2^4+x_3^4-18$ be polynomials in $\K[x_1, x_2,
x_3]^{B_3}$. The set $W(\phi, \f)$ contains all points in $\Kbar{}^3$ at which $f$ and all the 2-minors of 
\[
\jac =  \left(\begin{matrix}
4x_1^3 & 4x_2^3 & 4x_3^3 \\
2x_1x_2^2x_3^2 - 6x_1^2 & 2x_2x_1^2x_3^2 - 6x_2^2 & 2x_3x_1^2x_2^2 - 6x_3^2
\end{matrix} \right)
\] vanish. That is, the set $W(\phi, f)$ is equal to the zero set of the system 
\begin{equation} \label{eq1}
\begin{split}
x_1^4+x_2^4+x_3^4-18 &=  0 \\
8x_1x_2(x_1-x_2)(x_1+x_2)(x_1^2x_3^2 + x_2^2x_3^2-3) &=  0  \\
8x_1x_3(x_1-x_3)(x_1+x_3)(x_1^2x_2^2 + x_2^2x_3^2-3) &=  0 \\ 
8x_2x_3(x_2-x_3)(x_2+x_3)(x_1^2x_2^2 + x_1^2x_3^2-3) &=  0.
\end{split}
\end{equation}
The equations in \eqref{eq1} are not $B_3$-invariant and it can be shown that its solution set is of cardinality  $148$. %, 
%for example, by using  Gr\"{o}bner bases computation. 
 \end{example}

The problem of finding points in $W(\phi, \f)$ appears in many application areas
 such as polynomial optimization [5, 20 , 21, 27 ] and
 real algebraic geometry [2, 4, 6, 8, 22 , 29--31 ], especially in computing critical points for
 algebraic systems. If the Jacobian matrix of $\f$ with respect to
 $\X$ has full rank at any solution of $\f = 0$, then the algebraic
 set $V(\f)$ is smooth and $(n-s)$-equidimensional. In this case
 $W(\phi, \f)$ is the critical points set of $\phi$ restricted to
 $V(\f)$ by using Jacobian criterion~[11, Theorem 16.19].  

Computationally, exploiting  the invariance property allows us to
compute the set $W(\phi, \f)$ more efficiently. In particular we need to
compute only one point in the $B_n$-orbit of a fixed point rather than 
the entire orbit. Recall that the $B_n$-orbit of a point $\a$ is the
set $\{\tau(\a) \, : \, {\rm for \, all\,} \tau \in B_n\}$.    
\begin{example} Let $f$ and $\phi$ be polynomials defined in Example \ref{ex:import}. 
 Then instead
of computing the entire set $W(\phi, f)$ of 148 points, our algorithm
computes only a  representation for a set of size $8$.  \label{ex:import_2}
\end{example}
Note that the above problem is a generalization of solving
zero-dimensional  polynomial systems $f_1 = \dots =  f_n = 0$ when
$f_i$'s are in $\K[x_1, \dots, x_n]^{B_n}$.

%%%%%%%%%%%%%

\paragraph{Previous works.}  When $\phi$ is linear and $\f$ are any
polynomials in $\K[x_1, \dots, x_n]$, then there exists an algorithm [7, Section 14.2]  which computes the critical points set
$W(\phi, \f)$ using $d^{O(n)}$ operations in $\K$, where $d =
\max(\deg(\f))$. When $\f$ are all generic enough of degrees $d$,  the authors in [ 14, Corollary 3] 
use  Gr\"{o}bner basis techniques to show that computing $W(\phi, \f)$ requires   
\[
O\left( {{n+d_{\rm reg}}\choose{n}}^{\omega} + n\big( d^s(d-1)^{n-s}
  {{n-1}\choose {s-1}}^3\big)\right)  
\] operations in $\K$. Here $d_{\rm reg} = d(s-1) + (d-2)n+2$ and
$\omega$ is the exponent of multiplying 
two $(n \times n)$-matrices over $\K$. A generalization to 
systems with mixed degrees is later given in
[33].

When $\f$ and $\phi$ are polynomials invariant under the action of the
symmetric group $S_n$, the group of permutations of $\{1, \dots, n\}$,
then an algorithm in~[13] computes a  symmetric
representation of $W(\f, \phi)$. The runtime of this algorithm is
polynomial in $d^s, {{n+d} \choose {n}}$, and ${{n} \choose
  {s+1}}$. In Subsection~\ref{subsec"symm}, we will recall the
definition of symmetric representation and the result
in~[13], which is heavily used in this paper. 

Pioneering works in [23, 25] using symbolic homotopy techniques
give algorithms to compute a zero-dimensional parameterization to
represent the set $$V_p(\G, \f) := \{\x\in \Kbar{}^n 
\, : \, \f(\x) = 0 {\rm ~and~ } \rank(\G(\x)) < p\}.$$  Here $\G \in
\K[x_1, \dots, x_n]^{p \times q}$ and $\f = (f_1, \dots, f_s)$ is a
sequence of polynomials in $\K[x_1, \dots, x_n]$ with $p\le q$ and
$n=q-p+s+1$. The authors in [23] study two (total) degree measures for the
matrix $\G$, namely column degrees and row degrees. In some 
cases, for example in~[13], the polynomials in $\f$ and
$\G$ are sparse and it is observed that sparsity can  be exploited to get faster algorithms in polynomial system solving.
In this case, in~[25], the authors provide an algorithm to compute a representation for $V_p(\G, \f)$  which exploits the sparsity of the entries of $\f$ and $\G$.

When $\f$ and $\phi$ are $B_n$-invariant, then one can
take $\G$ as the Jacobian matrix of $(\f, \phi)$ with respect to
$(x_1, \dots, x_n)$. However, using the methods in [23, 25],  would not
exploit the invariance properties of $\f$ and $\phi$.  
Also when $\phi$ and $\f$ are $B_n$-invariant, then
one can use Gr\"{o}bner basis techniques to compute $W(\phi,
\f)$. However the invariance properties of the input polynomials  are lost
during the computation and again one cannot exploit this invariance.

%%%%%%%%%%%%% 

\paragraph{Our result.} 
Since  $\f = (f_1, \dots, f_s)$ and $\phi$ are  in $\K[\X]^{B_n}$,
there exist polynomials $\g = (g_1, \dots, g_s)$ and $\varphi$ in
$\K[\Z]^{S_n}$, where $\Z = (z_1, \dots, z_n)$ are new variables,  
such that $g_i(x_1^2, \dots, x_n^2) = f_i(\X)$ for all $i=1, \dots, n$
and $\varphi(x_1^2, \dots, x_n^2) = \phi(\X)$.  By the chain rule, 
$$
\jac_\X(\f, \phi) = \jac_\Z(\g, \varphi)({x_1^2, \dots, x_n^2})\cdot
\jac_\X(x_1^2, \dots, x_n^2). 
$$
The determinant of $\jac_\X(x_1^2, \dots, x_n^2)$ is
equal to $2^nx_1\cdots x_n$ and hence when $x_1, \dots, x_n$ are all
nonzero, then the matrix $\jac_\X(\f, \phi)$ is rank deficient if and only if
$\jac_\Z(\g, \varphi)$ is rank deficient. 

Thus if all $x_1, \dots, x_n$ are nonzero, then our problem reduced to a new
problem of computing a critical points set $W(\varphi, \g)$, a set which is finite if $W(\phi, \f)$ is finite. Thus we look for the critical points  of $\varphi$ restricted to the algebraic set $V(\g)$ but now with all $\g$ and $\varphi$ being invariant
under the action of the symmetric group $S_n$. This new problem can be
solved by using the results in~[13]. However, the set
$W(\phi, \f)$ may also contain some points with zero coordinates. Therefore, we also have to study the cases when some $x_i$'s
equal zero.  Note that if the degrees of $\f$ are at most $d$ (an even
positive integer), then those of  $\g$ are at most $d/2$. 

In this paper we introduce a notion called \emph{hyperoctahedral representation}
% in Subsection \ref{subsec:Hyperoctahedral representation} 
to represent a
$B_n$-invariant set. In addition we provide an algorithm to compute this representation with the running time polynomial in its size. 

\begin{theorem} Let $\f = (f_1, \dots, f_s)$ and $\phi$ be
  $B_n$-invariant polynomials in $\K[x_1, \dots, x_n]$. Assume that
  the set $W(\phi, \f) \subset \Kbar{}^{n}$ is finite.  
Then there exists a  randomized algorithm which takes $\f$ and $\phi$
as the inputs and returns a hyperoctahedral representation for
$W(\phi, \f)$  with the running time is polynomial in the size of the
output.  Moreover, the total number of points described by the output
is at most $n\,C$, where $C = (d/2)^s
{{n+d/2-1}\choose{n}}$. \label{thm:main}  
\end{theorem} 

Later in Theorem~\ref{thm:main_thm} we give a more precise estimate for the
runtime of the algorithm. We remark that the bound for the total
number of points described by the output is not really tight. We will
see in Section \ref{sec:expruns} that indeed, in our  experiments, the
number of points described by the output of our algorithm is much
smaller than $n\,C$. Theorem~\ref{thm:main} shows that the runtime of
our algorithm is always polynomial in the size of a  {hyperoctahedral 
  representation} of $W(\phi, \f)$.   We will also see in Section
\ref{sec:expruns} that, in practice, algorithms using Gr\"{o}bner
basis techniques are slower than the new algorithm given in this paper.   

Our algorithm is randomized in the sense that it makes random
choices of points which leads to a correct computations; those points
are not in certain Zariski closed sets of suitable affine spaces. In
this sense, our algorithm is of Monte  Carlo type, that is, it
returns the correct output with high probability, at least a fixed
value greater than 1/2. The error probability can be made to be
arbitrarily small by repeating the algorithm.

%%%%%%%%%%%%%  

\paragraph{Organization.} 
The rest of the paper is structured as follows. In
Section~\ref{sec:pre}, we recall some notations of
hyperoctahedral groups and some results in~[13]. In
Section~\ref{sec:some_important}, we introduce some data structures which
are used to define  hyperoctahedral representations for
$B_n$-invariant sets. Our main algorithm to find a  hyperoctahedral
representation for $W(\phi, \f)$ and its complexity can be found in
Section~\ref{sec:main_alg}. The paper ends with  a set of experimental
runs supporting the results in this paper. %An example
%illustrating the steps of our algorithm is added to {\color{red}{}
%Appendix~\ref{sec:exp}}  for the benefit of the reader. 

%%%%%%%%%%%%%%%%%%%%%%%%%%%%%%%%%%%%%%
%%%%%%%%%%%%%%%%%%%%%%%%%%%%%%%%%%%%%%
%%%%%%%%%%%%%%%%%%%%%%%%%%%%%%%%%%%%%%

\section{Preliminaries}
\label{sec:pre}
\subsection{Zero-dimensional parametrization}

 Let $S \subset \Kbar{}^n$ be a finite algebraic set. Then a
{\em zero-dimensional} {\em  parametrization} $\scrR = ((v, v_1, \dots,
v_n), \beta)$ of $S$ consists of polynomials $(v, v_1, \dots, v_n)$
in $\K[t]$ such that $v$ is monic and squarefree, all $v_i$'s are in
$\K[t]$ such that $\deg(v_i) < \deg(v)$, and $\beta = \beta_1t_1 +
\cdots + \beta_nt_n$ is a $\K$-linear form in $n$ variables $(t_1,
\dots, t_n)$, such that  
\begin{itemize}
    \item[(i)] $\beta(v_1, \dots, v_n) =  t\cdot v'$ with $v' =
      \frac{\partial v}{\partial t}$ and  
    \item[(ii)]  $S = Z(\scrR)$, where \[
    Z(\scrR) =  \left\{ \left( \frac{v_1(\tau)}{v'(\tau)}, \dots,
        \frac{v_n(\tau)}{v'(\tau)}\right) \, : \, v(\tau) =
      0\right\}. 
    \]
\end{itemize}
This representation was introduced in~[24,26] by Kronecker and Macaulay, and has been used
in many algorithms, see for examples [1, 16--19, 28] and references therein.  
  The representation allows us to use univariate polynomials to represent the output of our algorithm.

The reason why we use this parametrization, with $v'$ in the
denominator, goes back to [1, 19, 28]. This parametrization allows us
to control the bit-size of the coefficients when $\K = \mathbb{Q}$,
the field of rational numbers, by using bounds in
[10, 32]. The same phenomenon holds
with $\K = k(x)$, the field of rational functions over a field $k$,
 where we want to control the degrees of the numerators and
denominators of the coefficients of  $\scrR$.

%%%%%%%%%%%%%%%%%%%%%%%%%%%%%%%%%

\subsection{Hyperoctahedral groups}
\label{subsec:Hyper_groups} 
Given a positive integer $n$, the hyperoctahedral group $B_n$ is the
group of all signed permutations. For $n \ge 2$, $B_n$ is a subgroup
of the orthogonal group in dimension $n$ consisting of maps $\tau$
of the form  
\[
\tau(e_1, \dots, e_n) = (\pm e_{\sigma(1)}, \dots, \pm e_{\sigma(n)}),
\] where $\sigma$ belongs to the symmetric group $S_n$.

Equivalently, $B_n$ is the group of all signed permutation matrices, a group with
order $2^n n!$. We can also represent the group $B_n$ as
the semi-direct product $\Delta_n \rtimes S_n$, where the symmetric
group $S_n$ is identified with the group of $n\times n$ permutations
matrices, and $\Delta_n$ is the group of all diagonal matrices of size
$n$ with entries being $\pm 1$.  The group $B_n$ is a finite 
reflection group and is the symmetric group of the $k$-dimensional
hypercube.  A set $W$ of points is called $B_n$-invariant if $\tau(\bm 
e) \in W$ for all $\bm e \in W$ and $\tau \in B_n$.  

\begin{example} Consider $n=3$. The group $B_3$ is the full symmetry 
  group of the octahedron. There are two types of reflection planes:
  the first one is defined by $x_i=0$ for $1 \le  i \le 3$ and the
  second one is defined by $x_i = x_k$ for $1\le i < k \le 3$.  
In addition, the set $W = \{(-1, -2), (-2,-1), (-1,2), (2,-1), (1,
-2), (-2,1), (1,2), (2,1)\}$ is \newline
$B_2$-invariant. \label{ex:Bn_invariant} 
\end{example}

Given a positive integer $n$, we denote by $E_n : \Kbar{}^n \rightarrow \Kbar{}_{\ge 0}^n$  the mapping
{
\[\e =  (e_1, \dots, e_n) \mapsto (e_1^2, \dots, e_n^2). \]}
For a nonempty set $W \subset \Kbar{}^n$, we denote by $W_{\rm rem}$ the image of  $W$ under $E_n$.  It is trivial to observe that $W_{\rm rem}$ is
$S_n$-invariant
%where $S_n$ is the symmetric group, i
if $W$ is a nonempty
$B_n$-invariant set.  

\begin{example} Let $W$ be the set defined in
  Example~\ref{ex:Bn_invariant}. Then $W_{\rm rem} = \{(1,4),
  (4,1)\}$. On the other hand, if we consider $W = \{(1,0), (-1, 0), 
  (0,1), (0,-1)\}$, then $W_{\rm rem} = \{(1,0), (0,1)\}$. Both of
  these $W_{\rm rem}$ are $S_2$-invariant. 
\end{example}

Let $\K[x_1, \dots, x_n]$ be a multivariate polynomial ring with
coefficients in $\K$. The hyperoctahedral group $B_n$ acts on $\K[x_1,
\dots, x_n]$ by permuting the $x_i$'s and changing the signs of an
arbitrary number of variables. A polynomial $p$ in $\K[x_1, \dots,
x_n]$ is called a $B_n$-invariant or hyperoctahedral polynomial if
$$\tau(p):= p(\tau(x_1, \dots, x_n)) = p \ {\rm for \ all} \ \tau \in B_n.$$ We
denote by $\K[x_1, \dots, x_n]^{B_n}$ the ring of all $B_n$-invariant
polynomials.  

Let $p \in \K[x_1, \dots, x_n]^{B_n}$. Since $\Delta_n \subset B_n$,
 there exists a unique polynomial $g \in \K[z_1, \dots, z_n]$,
where $(z_1, \dots, z_n)$ are new variables, such that $g(x_1^2,
\dots, x_n^2) = p$. Since $S_n \subset B_n$, then $g$ is a symmetric
polynomial in $\K[z_1, \dots, z_n]$. Therefore, there exists a unique
polynomial $h \in \K[e_1, \dots, e_n]$, where $(e_1, \dots, e_n)$ are
new variables, such that   $$h(\eta_1(z_1, \dots, z_n), \dots,
\eta_n(z_1, \dots, z_n)) = g$$ with $\eta_k(\cdot)$ being the
elementary symmetric function of degree $k$ in $(z_1, \dots, z_n)$. In
other words, $(\eta_1(x_1^2, \dots, x_n^2), \dots, \eta_n(x_1^2,
\dots, x_n^2))$ is a set of generators for $\K[x_1, \dots,
x_n]^{B_n}$.   

%%%%%%%%%%%%%%%%%%%%%%%%%%%%%%%%%%%%%%

\subsection{Critical points for symmetric sets}
\label{subsec"symm}
%Given a positive integer $m$, let $S_m$ be the symmetric group. 
In
this subsection we summarize the results in~[13] used to
compute critical points defined by polynomials which are invariant by
the action of the symmetric group $S_m$ for $m$ a positive integer.

%Let $m$ be a positive integer. 
A set of positive integers $\lambda =
(m_1^{k_1}\, \dots $ $m_r^{k_r})$, with $m_1 < \cdots < m_r$, is
called a partition of %a positive integer 
$m$ if $m_1k_1 + \cdots + m_rk_r = m$. The
number $k = k_1 + \cdots + k_r$ is called the length of the partition
$\lambda$. The notation $(m_1^{k_1}\, \dots m_r^{k_r})$ can be
interpreted as the ordered list $(m_1, \dots, m_1, \dots,$ $m_r, \dots, 
m_r)$, with each $m_i$ appeared $k_i$ times. We denote by $U_\lambda$
and $U_{\lambda, \rm dist}$ the sets of points in $\Kbar{}^m$ as   
\begin{multline} \label{eq:subset_par} 
U_\lambda := \Big\{\big(\underbrace{a_{1,1}, \dots, a_{1,1}}_{m_1},
\dots, \underbrace{a_{1, k_1}, \dots, a_{1, k_1}}_{m_1}, \dots, 
\underbrace{a_{r,1}, \dots, a_{r, 1}}_{m_r}, \\ \dots,  \underbrace{a_{r,
    k_r}, \dots, a_{r, k_r}}_{m_r}\big) \in \Kbar{}^{m} : a_{i,j} \in \Kbar {\rm ~for~all~} i,j\Big\} 
\end{multline}
and 
$
U_{\lambda, \rm dist} := \Big\{\bm a \in U_\lambda  \, : \,  a_{i, j}
{\rm ~are~ pairwise~distinct~} {\rm ~for~all~} i,j \Big\}. 
$

For a partition $\lambda = \partition$ of $m$ of length $k$, we denote
by $F_\lambda : U_\lambda \rightarrow \Kbar{}^k$ the mapping  
\[
F_\lambda : \a \in U_\lambda \mapsto \big(\eta_1(a_{i, 1}, \dots,
a_{i, k_i}), \dots, \eta_{k_i}(a_{i, 1}, \dots, a_{i, k_i})\big)_{1
  \le i \le r}, 
\] where for $1 \le i \le r$ and $1 \le \ell \le k_i$,
$\eta_\ell(a_{i, 1}, \dots, a_{i, k_i})$ is the  $\ell$-th elementary
symmetric function in $(a_{i, 1}, \dots, a_{i, k_i})$. 
Let $Q$ be a $S_m$-invariant subset of $\Kbar{}^m$. We write  
\[
Q_\lambda = S_m(Q\cap U_{{\lambda, \rm dist}}) \quad {\rm and} \quad
Q_\lambda' = F_\lambda(Q\cap U_{{\lambda, \rm dist}}),  
\] where $S_m(Q\cap U_{{\lambda, \rm dist}})$ is the orbit of $Q\cap
U_{{\lambda, \rm dist}}$ under the action of $S_m$. The set
$Q_\lambda'$ is a compression of $Q_\lambda$. 

\begin{definition}
Let $Q$ be a finite $S_m$-invariant set in $\Kbar{}^m$. A
\emph{symmetric representation} of $Q$ is a sequence $(\lambda_i,
\scrR_i)_{1 \le i \le L}$, where the $\lambda_i$'s are all the partitions
of $m$ for which $Q_{\lambda_i}$ is not empty, and where for each $i$,
$\scrR_i$ is  a zero-dimensional parametrization of $Q'_{\lambda_i}$. 
\end{definition}

\begin{theorem}$[13]$ \label{them:FLSSV2021}
Let $\q = (q_1, \dots, q_s)$ and $\varphi$ be symmetric polynomials in
$\K[x_1, \dots, x_m]$ with $s < m$. Assume further that the set
$W(\varphi, \q) := \{ \x \in \Kbar{}^m  :  \q(\x) = 0 {\rm ~and~}
\jac_\X(\q, \varphi) < s+1\}$ is finite. 

Then there exists a randomized algorithm  ${\sf
  Symmetric\_Represen}$- ${\sf tation}(\q, \varphi)$ which takes as
input $\q$ and $\varphi$, and returns a symmetric representation
$\scrR  = (\lambda_i, \scrR_i)_{1 \le i \le   L}$ for $W(\varphi, \q)$
whose runtime is    
\[
\softO\big(\delta^2(c+e^5)m^9 \rho \big) \subset \left( \delta^s
  {{m+\delta}\choose d} {{m} \choose {s+1}}\right)^{O(1)} 
\] operations in $\K$, where $\delta = \max(\deg(\q), \deg(\varphi))$,
$c = \delta^s {{m+\delta-1}\choose{m}}$, $e = m(\delta+1){{m+\delta}
  \choose{m}}$, and $\rho = m^2{{m+\delta}\choose {\delta}} + m^4
{{m}\choose {s+1}}$.  

%Moreover, t
The total number of points described by the output is at
most $c$. 
\end{theorem}

%%%%%%%%%%%%%%%%%%%%%%%%%%%%%%%%%%%%%%
%%%%%%%%%%%%%%%%%%%%%%%%%%%%%%%%%%%%%%
%%%%%%%%%%%%%%%%%%%%%%%%%%%%%%%%%%%%%%

\section{Hyperoctahedral
  representation} 
\label{sec:some_important} 

%Consider a positive integer $n$. 
Let $B_n$ be the hyperoctahedral
group. In this section, we  describe some data structures of finite sets and  introduce the notion of a hyperoctahedral
representation to represent a $B_n$-invariant set. 

%%%%%%%%%%%%%%%%%%%%%%%%%%%%%%%%

\subsection{Data structures}

Let $m$ and $m'$ be  positive integers such that $m' \le m \le n$ and let  $\lambda = (m_1^{k_1}\, \dots m_r^{k_r})$ and $\lambda'=
(v_1^{u_1} \, \dots \, v_s^{u_s})$ be partitions of $m$ and $m'$,
respectively. We say $\lambda'$ is an {\em extended refinement} of
$\lambda$, and  denote it by $\lambda' >_{\rm ext} \lambda$, if
either $m'<m$ or $\lambda'$  is the union of some partitions
$(\lambda_{i,j})_{1 \le i \le r, 1 \le j \le k_i}$, where
$\lambda_{i,j}$ is a partition of $v_i$ for all $i,j$. This 
is an extension of the notion of {\em refinement order} which is
defined in [13, Section~2.1].  
 \begin{example}
   We have $ (1^2 \, 2) <_{\rm ext} (1 \, 3) <_{\rm ext} (3)$ with
   $(3)$ being a partition of $m'=3$ and $(1^2 \, 2)$ and $(1 \, 3)$ 
    being partitions of $m=4$.
  \end{example}
Let $\lambda = (m_1^{k_1}\, \dots, m_r^{k_r})$ be a partition of
$m$ of length $k$. Recall that $U_\lambda$ denotes a set of points in
$\Kbar{}^m$, defined previously in \eqref{eq:subset_par}, as   
\begin{multline*}
U_\lambda := \Big\{\big(\underbrace{a_{1,1}, \dots, a_{1,1}}_{m_1},
\dots, \underbrace{a_{1, k_1}, \dots, a_{1, k_1}}_{m_1}, \dots, 
\underbrace{a_{r,1}, \dots, a_{r, 1}}_{m_r}, \\ \dots,  \underbrace{a_{r,
    k_r}, \dots, a_{r, k_r}}_{m_r}\big) \in \Kbar{}^{m} : a_{i,j} \in \Kbar {\rm ~for~all~} i,j\Big\}.
\end{multline*}
We define  a subset $U_{\lambda, \rm dist,0}$  of $\Kbar{}^m$
as $$U_{\lambda, \rm dist,0} := \{\bm a \in U_\lambda : a_{i, 
  j} \ne 0 ~{\rm are~ pairwise~ distinct~} %and~} a_{i,j}\ne 0 ~
  \forall~
i,j\}.$$ The set $U_{\lambda,   \rm dist,0}$ is a restriction of
$U_{\lambda, {\rm dist}}$ on the open set given by  $\{a_{i, j} \ne 0 {\rm 
  ~for~all~ } i, j\}$, where $U_{\lambda, {\rm dist}}$ is a subset of
$\Kbar{}^m$ containing all points in $U_\lambda$ with pairwise
distinct coordinates. 

Let $m$ and $n$ be positive integers such that $m \le n$ and denote the zero vector of length $n-m$ by $\bm 0_{n-m}$. Define  
\begin{multline}
U_{\lambda, [n-m]} := \{(\bm a, \bm 0_{n-m}) \, : \, \bm a \in
U_\lambda\} \, \,\,  {\rm and} \\\,\,\, U_{\lambda, [n-m], \rm dist}
:= \{(\bm a, \bm 0_{n-m}) \, : \,\, \bm a \in U_{\lambda, \rm dist,
  0}\}.  
\end{multline} The set $U_{\lambda, [n-m]}$ is the union of all
$U_{\lambda', [n-m'], \rm dist}$ for all $\lambda'  >_{\rm ext}
\lambda$.  

\begin{example} 
Let us consider a partition $\lambda = (1^3\,2^2)$ of
$m=7$. Then $m_1=1, m_2=2, k_1 = 3, k_2=2$, and the length of
$\lambda$ is $k=5$. The set $U_{\lambda, [2]}$ contains all
points of the form $$\bm a =(a_{1,1}, a_{1,2}, a_{1,3}, a_{2,1},
a_{2,1}, a_{2,2}, a_{2,2}, 0, 0)$$ 
while  $U_{\lambda, [2], \rm dist}$ is a subset of $U_{\lambda, [2]}$
with $a_{i,j}$ being distinct and nonzero numbers.    
\label{ex:Fbml}
\end{example}

To any point $\bm c$ in $\Kbar{}^n$, the \emph{type} associated to
$\bm c$ is a pair $(\lambda, [n-m])$ of the unique partition $\lambda
= \partition$ of $m$ and an integer $n-m$ such that there exists
$\tau \in B_n$ for which $\tau(\bm c)$ lies in $U_{\lambda, [n-m],
  \rm dist}$. Since all points in an orbit have the
same type, we can then define the type of an orbit 
as the type of any points in that orbit. The size of any orbit of type
$(\lambda, [n-m])$ is    
  \begin{multline}\label{eq:zeta111}
\zeta_{\lambda,
  [n-m]} := {{n} \choose{m_1, \dots, m_1, \dots, m_r, \dots, m_r,
    n-m}}  \\= \frac{n!}{m_1!^{k_1} \cdots m_r!^{k_r} \, (n-m)!}. 
\end{multline}
 Note that while all points in $U_{\lambda, [n-m], \rm  dist}$ have
 type $(\lambda, [n-m])$, it is not correct that all points in
 $U_{\lambda, [n-m]}$ are of type $(\lambda, [n-m])$.     

For $m$ and $\lambda = (m_1^{k_1}\, \dots, m_r^{k_r})$ as above, we 
define $J_\lambda := \{(\bm \varepsilon, 0) \in \Kbar{}^{k+1} \, :
\,\bm \varepsilon \in \Kbar{}^k\}$ and the mapping
$
L_\lambda : U_{\lambda, [n-m]} \rightarrow J_\lambda
$ as
\begin{equation} \label{eq:llambda}
    \c \mapsto \big(\eta_1(a_{i,1}, \dots, a_{i, k_i}), \dots,
    \eta_{k_i}(a_{i,1}, \dots, a_{i, k_i}), 0\big)_{1 \le i \le r},  
\end{equation}
where $\eta_\ell(\cdot)$ is the degree $\ell$ elementary symmetric
function. The last zero coordinate represents the compression 
of the zero vector ${\bf 0}_{n-m}$ component in $U_{\lambda, [n-m]}$
into a single zero coordinate. One can think of the map $L_{\lambda}$ as
a compression of an orbit $\mathcal{O}$ to  a single point
$L_\lambda(\mathcal{O} \cap U_{\lambda, [n-m]}) =
L_\lambda(\mathcal{O} \cap U_{\lambda, [n-m],    \rm dist})$. 

\begin{example}\label{ex:mapLlambda} Continuing with
  Example~\ref{ex:Fbml}, we have
$L_\lambda(a_{1,1}, a_{1,2},$ $a_{1,3}, a_{2,1}, a_{2,1}, a_{2,2},
a_{2,2}, 0, 0) = (a_{1,1}+a_{1,2}+a_{1,3}, a_{1,1}a_{1,2} +
a_{1,1}a_{1,2}+a_{1,2}a_{1,3}, a_{1,1}a_{1,2}a_{1,3}, a_{2,1}+a_{2,2},
a_{2,1}a_{2,2}, 0)$.  
\end{example} 

Let $\lambda$ be a partition and $L_\lambda$ be the mapping as
above. The map  $L_\lambda$ is onto: for $\bm \vartheta = (\bm
\varepsilon, 0) = (\varepsilon_{1,1}, \dots, \varepsilon_{r, k_r}, 0)$
in $\Kbar{}^{k+1}$, we can find a point $(\bm a, {\bm 0}_{n-m})$ in
the preimage $L_\lambda^{-1}(\bm \vartheta)$ by finding the roots
$a_{i,1}, \dots, a_{i, k_i}$ of the
polynomial \begin{equation}\label{eq:vieta} 
q_i(t) = t^{k_i} - \varepsilon_{i, 1}t^{k_i-1} + \cdots +
(-1)^{k_i}\varepsilon_{i, k_i} 
\end{equation}
for all $i=1, \dots, r$. Therefore, we can define the preimage
$L_\lambda^{-1}(\bm \vartheta)$ of any point $\bm \vartheta = (\bm
\varepsilon, 0)$ in $J_\lambda$ and hence we can define
$L_\lambda^{*}(\bm \vartheta) := B_n(\c)$ as the orbit of any point
$\c \in L_\lambda^{-1}(\bm \vartheta)$. For any finite set $X \subset
J_\lambda$, we denote by $L_\lambda^{*}(X) := \{L_\lambda^{*}(\bm
\vartheta) \, : \, \bm \vartheta \in X\}$.

Notice that any point $(\alpha_1, \dots, \alpha_\ell) \in
\Kbar{}^\ell$, for some integer $\ell$, has $\alpha_i \ne 0$ for all 
$i$ if and only if $\eta_\ell(\alpha_1, \dots, \alpha_\ell) =
\alpha_1\cdots \alpha_\ell$ is nonzero. The image
$L_\lambda(U_{\lambda, [n-m], \rm dist})$ of any point of type
$\lambda$ is an open set $O_{\lambda, [n-m]} \subsetneq
\Kbar{}^{k+1}$, defined by the polynomials $q_i$ in~\eqref{eq:vieta}
with conditions saying that these polynomials are squarefree, pairwise
coprime, and are not divided by $t$. For any point $\bm \vartheta \in
\Kbar{}^{k+1} \setminus O_{\lambda, [n-m]}$, the orbit
$L_\lambda^{*}(\bm \vartheta)$ may not have type $(\lambda, [n-m])$,
it is of type $(\lambda', [n-m'])$ for some partition $\lambda' >_{\rm
  ext} \lambda$ of $m' \le m$.

\begin{example} \label{ex:manythings}
Let us continue with Example~\ref{ex:mapLlambda}. Any point $\bm \vartheta$
in $J_\lambda$ has the form $\bm \vartheta = (\bm \varepsilon, 0) =
(\varepsilon_{1,1}, \varepsilon_{1,2}, \varepsilon_{1,3},
\varepsilon_{2,1}, \varepsilon_{2,2}, 0)$ in $\Kbar{}^6$. The
polynomials $q_i(t)$ defined in~\eqref{eq:vieta} are $q_1(t) =
t^3-\varepsilon_{1,1}t^2 + \varepsilon_{1,2}t - \varepsilon_{1,3}$ and
$q_2(t) = t^2 - \varepsilon_{2,1}t+\varepsilon_{2,2}$. The open set
$O_{\lambda, [2]}$ is defined by $\varepsilon_{1,3}, \varepsilon_{2,2}
\ne 0$, ${\rm gcd}(q_1, q_2) =1$, and $q_1$ and $q_2$  both being
squarefree.  

The point $\bm \vartheta = (0, -7, 6,3,2,0)$ is in $O_{\lambda, [2]}$
since $q_1 = t^3+7t-6$ and $q_2 = t^2 - 3t+2$ are coprime, squarefree,
and these polynomials are not in the form $t^d p(t)$ for some positive
integer $d$ and polynomial $p$. The orbit $L_\lambda(0, -7, 6,3,2,0)$
contains all permutations of $(-1, -2, 3, 1,1,2,2,0,0)$.  
 
Finally, we consider the point $\bm \vartheta = (\bm \varepsilon, 0) =
(3,2,0,4,3,0)$. Then $q_1(t) = t^3-3t^2+2t$ and $q_2(t) = t^2-4t+3$,
from which we can conclude that  $(3,2,0,4,3,0)$ is not in
$O_{\lambda, [2]}$.  
 The preimage $L_\lambda^{-1}(\bm \vartheta)$ is $\bm a =
 (1,2,0,1,1,3,3,0,0)$ and the orbit $L_\lambda^{*}(\bm \vartheta)$ is
 a subset of $\Kbar{}^9$ containing all permutations of the point $\bm 
 a$. The orbit  $L_\lambda^{*}(\bm \vartheta)$ 
 has type $(\lambda', 3)$, where $\lambda' = (1^1\, 2^1\, 3^1)$ is a
 partition of $m'=6$, which can be seen by re-arranging $\bm a$ as
 $\bm a' = (2,3,3,1,1,1,0,0,0)$.  The partition $\lambda'$ is an
 extended refinement of $\lambda$ as $m'<m$.  
\end{example}

Given a point $\bm \vartheta \in J_\lambda$, we will need an algorithm 
to recover the type $(\lambda', [n-m'])$  of the orbit
$L_\lambda^*(\bm \vartheta)$. The following algorithm is an extension
of the ${\sf Type\_Of\_Fiber}$ algorithm 
in~[13, Lemma 2.11].  
\begin{lemma} There exists an algorithm  ${\sf
    Type\_Of\_Fiber\_Extended}$ $(\lambda, n-m, \bm \vartheta)$ which
  takes a partition $\lambda = \partition$ of m of length $k$, a
  non-negative integer $n-m$, and a point $\bm \vartheta \in
  \Kbar{}^{k+1}\cap J_\lambda$,  and returns a partition $\lambda'$ of
  some positive number $m'\leq n$ of length $k'$, an integer $n-m'$, and a
  tuple $\bm b \in \Kbar{}^{k'+1}$ such that the orbit
  $L_\lambda^*(\bm \vartheta)$ has type $(\lambda', [n-m'])$ and
  $L_{\lambda'}(\mathcal{O} \cap U_{\lambda', [n-m'], \rm dist}) =
  \{\bm b\}$. The algorithm requires $\softO(m)$ operations in $\K$. 
\label{lemma:fiber_type}
\end{lemma}

\begin{proof}
Let us write $\bm \vartheta = (\bm \varepsilon, 0) =
(\varepsilon_{1,1}, \dots, \varepsilon_{r, k_r}, 0)$. Any point in
$L_\lambda^{-1}(\bm \vartheta)$ is a permutation of the point 
\begin{multline*}
\bm c =  \Big\{\big(\underbrace{c_{1,1}, \dots, c_{1,1}}_{m_1}, \dots,
\underbrace{c_{1, k_1}, \dots, c_{1, k_1}}_{m_1}, \\ \dots,
\underbrace{c_{r,1}, \dots, c_{r, 1}}_{m_r}, \dots, \underbrace{c_{r,
    k_r}, \dots, c_{r, k_r}}_{m_r}, \underbrace{0, \dots,
  0}_{n-m}\big)\Big\}, 
\end{multline*} where, for $i=1, \dots, r$, $(c_{i, 1}, \dots, c_{i,
  k_i})$ are solutions of $q_i(t)$, with $q_i(t)$ being the polynomial
defined in~\eqref{eq:vieta}.

Notice that any point $(\alpha_1, \dots, \alpha_\ell) \in
\Kbar{}^\ell$, for some integer $\ell$, has $\alpha_i \ne 0$ for all
$i$ if and only if $\eta_\ell(\alpha_1, \dots, \alpha_\ell) =
\alpha_1\cdots \alpha_\ell$ is nonzero. Then finding the type of $\bm
c$ is equivalent to finding the repetitions among the $c_{i,j}$'s and the
zero coordinates in ($\varepsilon_{i, k_i})_{1\le i \le r}$. This can
be done by first computing  the product 
\begin{multline*}
q =(t^{k_1} - \varepsilon_{1, 1}t^{k_1-1} + \cdots +
(-1)^{k_1}\varepsilon_{1, k_1})^{m_1}\\ \cdots (t^{k_r} -
\varepsilon_{1, 1}t^{k_r-1} + \cdots + (-1)^{k_r}\varepsilon_{r,
  k_r})^{m_r}.
\end{multline*} Its factorization is then of the form $q =
p_1^{v_1} \cdots  p_s^{v_s}\, t^d$, with $\deg(p_i) = u_i$, $v_1 <
\cdots < v_s$, all $p_i$'s squarefree,  pairwise coprime, and all not divisible by $t$. In addition
$u_1v_1+\cdots+u_sv_s+d=m$. Then $\bm c$ has type  
$(\lambda', [n-m+d])$, with $\lambda' = (v_1^{u_1}\,\dots\,v_s^{u_s})
>_{\rm ext} \lambda$. If we write $p_i = t^{u_i} - b_{i, 1}t^{u_i-1} +
\cdots + (-1)^{u_i}b_{i, u_i}$ for $i=1, \dots, s$, then the output of
our ${\sf Type\_Of\_Fiber\_Extended}$ algorithm is $(\lambda', n-m+d,
\bm b)$, where $\bm b = (b_{1,1}, \dots, b_{s, u_s}, 0)$.  

Similar to the algorithm in~[Lemma~2.11, 13],  we can
compute the form $p_1^{v_1} \cdots  p_s^{v_s}\, t^d$ of $q$ by using
subproduct tree techniques~[15, Chapter~10] and fast
gcd~[15, Chapter~14] using a total of $\softO(m)$
operations in $\K$.  
\end{proof}

\begin{example} Let us continue Example~\ref{ex:manythings} with $\bm
  \vartheta = (\bm \varepsilon, 0) = (3,2,0,4,3,0)$. Then $q(t) =
  (t^3-3t^2+2t)^1 (t^2-4t+3)^2$, which can be factorized as  
\[
q(t) = (t-2)^1(t-3)^2(t-1)^3\,t^1. 
\] From this, we will have $s = 3, v_1=1, v_2 = 2, v_3 =3, p_1(t)
= t-2, p_2(t) = t-3, p_3(t) = t-1$, and $d=1$. Therefore $\lambda' =
(1^1\, 2^1\, 3^1)$ and the output of ${\sf
  Type\_Of\_Fiber\_Extended}((1^3\,2^2), 2, \bm \vartheta)$ is
$(\lambda', 3,$ $(2,3,1, 0))$, the last being equal to
$L_{\lambda'}(2,3,3,1,1,1,0, 0,0)$.  
\end{example}

%%%%%%%%%%%%%%%%%%%%%%%%%%%%%%%%%%%%%%%%%%
\subsection{Hyperoctahedral representation}

\label{subsec:Hyperoctahedral representation}
We now give a representation for an $B_n$-invariant set in 
$\Kbar{}^n$. For a nonempty finite subset $W$ of $\Kbar{}^n$, we let $W_{\rm
  rem} \subset \Kbar{}^n$ be the set obtained from $W$ as in
Subsection~\ref{subsec:Hyper_groups}. We recall it here for the readers convenience.

We let $E_n : \Kbar{}^n \rightarrow \Kbar{}_{\ge 0}^n$  be the mapping
\[\e =  (e_1, \dots, e_n) \mapsto (e_1^2, \dots, e_n^2). \]
For a nonempty set $W \subset \Kbar{}^n$, we denote by $W_{\rm rem}$ the image of  $W$ under $E_n$. Then $W_{\rm rem}$ is $S_n$-invariant if $W$ is
$B_n$-invariant. %, where $S_n$ is the symmetric group of order $n!$.  
For any point $\bm w = (w_1, \dots, w_n) \in \Kbar{}_{\ge 0}^n$, we can define  the preimage of $\bm w$ as
\begin{equation}
    \label{eq:E*} E_n^{*}((w_1, \dots, w_n)) := \{(\pm (w_1)^{1/2},
    \dots, \pm(w_n)^{1/2})\},
\end{equation}
a set consisting of at most $2^n$ points. 

Let $W \subset \Kbar{}^n$ be a $B_n$-invariant set. 
With all notations as above, the cardinality of the set $L_\lambda(W_{\rm
  rem} \cap U_{\lambda, [n-m], \rm dist})$, where $L_\lambda$ is
defined in~\eqref{eq:llambda}, is smaller than that of the orbit
$B_n(W_{\rm   rem} \cap U_{\lambda, [n-m], \rm dist})$ of the set
$W_{\rm rem} \cap U_{\lambda, [n-m], \rm dist}$ by a
factor $$\Gamma_{\lambda, [n-m]} := \zeta_{\lambda, [n-m]} \cdot (~k_1!
\cdots k_r! (n-m)!~),$$ where $\zeta_{\lambda, [n-m]}$ is defined in
\eqref{eq:zeta111}. In addition, the cardinality of the set
$B_n(W_{\rm rem} \cap U_{\lambda,   [n-m], \rm dist})$  is smaller
than that of $W$ by a factor of $\gamma_{_{\lambda, [n-m]}}:= 2^k$,
where $k$ is the length of the  partition $\lambda$. Altogether, we
have the following definition.

\begin{definition} Let $W$ be a finite $B_n$-invariant set in
  $\Kbar{}^n$. A {\em hyperoctahedral representation} of $W$ is a set
  $(\lambda_i, \scrQ_i, {\bm 0}_{n-m})_{1 \le i \le M, 0 \le m \le
    n}$, where $\lambda_i$'s are all partitions of $m \in \{0, \dots,
  n\}$ such that the set $B_n(W_{\rm rem}\cap U_{\lambda_i, [n-m], \rm
    dist})$ is non-empty and for each $i$, $\scrQ_i$ is a
  zero-dimensional parametrization of $L_{\lambda_i}(W_{\rm rem}\cap
  U_{\lambda_i, [n-m], \rm dist})$.  
\end{definition}

\begin{example} Let us consider a $B_3$-invariant set $W = W_1 \cup
  W_2 \cup W_3$ of 44 points with $W_1 = \{\sigma(\pm 1, \pm 2, 0) :
  \sigma\in S_3\}, W_2 = \{\sigma(\pm 3, \pm 3, \pm 4) : \sigma\in
  S_3\}$, and $W_3 = \{\sigma(\pm 5, \pm 5, \pm 5) : \sigma\in
  S_3\}$. The cardinalities of the sets $W_1, W_2$, and $W_3$ are respectively
  $24$, $12$, and $8$. Then the set $W_{\rm rem} = W_{1, \rm rem} \cup W_{1,
    \rm rem} \cup W_{3, \rm rem}$ consisting of 10 points with $W_{1,
    \rm rem} = \{\sigma(1,4,0) :  \sigma \in S_3\}$ of cardinality 6,
  $W_{2, \rm rem} = \{\sigma(9,9,16) :  \sigma \in S_3\}$ of cardinality
  3, and $W_{3, \rm rem} = \{(25,25,25)\}$ of cardinality 1.  

With $m=2$ and $\lambda_1 = (1^2)$, we have $L_{\lambda_1}(W_{\rm
  rem} \cap  U_{\lambda_1, [0], \rm dist}) = L_{\lambda_1}(W_{1, \rm
  rem} \cap  U_{\lambda_1, [1], \rm dist}) = \{(1,4)\} \subset
\Kbar{}^2$, $\Gamma_{\lambda_1, [1]} = 6$, and $\gamma_{\lambda_1,
  [1]} = 4$. With  $m=3$ and $\lambda_2 = (1^1\, 2^1)$, we have
$L_{\lambda_2}(W_{\rm rem} \cap  U_{\lambda_2, [0], \rm dist}) =
L_{\lambda_2}(W_{2, \rm rem} \cap  U_{\lambda_2, [0], \rm dist}) =
\{(9,16)\} \subset \Kbar{}^2, \Gamma_{\lambda_2, [1]} = 3$, and
$\gamma_2 = 4$. Finally with $m = 3$ and $\lambda_3 = (3^1)$, we 
have $L_{\lambda_3}(W_{\rm rem} \cap  U_{\lambda_3, [0], \rm dist}) =
L_{\lambda_3}(W_{3, \rm rem} \cap  U_{\lambda_2, [0], \rm dist}) =
\{(25)\} \subset \Kbar, \Gamma_{\lambda_2, [1]} = 1$, and
$\gamma_{\lambda_3, [0]} = 2$.  

A hyperoctahedral representation of $W$  consists of $(\lambda_1,
\scrR_{\lambda_1},$ $ [1])$, $(\lambda_2, \scrR_{\lambda_2}, [0])$,
and $(\lambda_3, \scrR_{\lambda_3}, [0])$ with $Z(\scrR_{\lambda_1}) =
\{(1,4)\}, Z(\scrR_{\lambda_2}) $ $=  \{(9, 16)\}$,
$Z(\scrR_{\lambda_3}) =  \{(25)\}$, $\lambda_1 = (1^2)$ being a partition
of $m=2$, and $\lambda_2 = (1^1\, 2^1)$ and $\lambda_3 = (3^1)$ being
partitions of $m=3$.  
\end{example}

At some situations, we may have to deal with the following problem. We
have a zero-dimensional representation of a finite set $G \subset
\Kbar{}^{k+1}$ and we want to compute a hyperoctahedral representation
of $E_n^{*}(L_\lambda^*(G))$ for some partition of length $k$. Note
that some points in $G$ may correspond to orbits having type $\lambda'
>_{\rm ext} \lambda$ as we have seen in Example~\ref{ex:manythings}.  
\begin{lemma} Let $\lambda$ is a partition of $m$ of length $k$ and
  $\scrR$ be a zero-dimensional parametrization of a set $G$ in
  $\Kbar{}^{k+1}\cap J_\lambda$.

Then there exists a randomized algorithm  ${\sf Decompose\_Extended}$
$(\lambda, \scrR, n-m)$, which takes as input $\lambda, G$, and $n-m$,
and outputs a hyperoctahedral representation of
$E_n^{*}(L_\lambda^*(G))$. The runtime of the algorithm is
$\softO(\delta^2\, m)$ operations in $\K$ with $\delta = \deg(\scrR) =
|G|$. \label{lemma:decom_ext}  
\end{lemma}

\begin{proof}
Our ${\sf Decompose\_Extended}$ algorithm is a slight modification of
the ${\sf Decompose}$ algorithm in [13, Lemma~2.16] where, 
instead of using ${\sf Type\_Of\_Fiber}$, we use the ${\sf
  Type\_Of\_Fiber\_Extended}$ algorithm in
Lemma~\ref{lemma:fiber_type}.  
\end{proof}

%%%%%%%%%%%%%%%%%%%%%%%%%% 
%%%%%%%%%%%%%%%%%%%%%%%%%% 
%%%%%%%%%%%%%%%%%%%%%%%%%% 

\section{Computing the critical points}
\label{sec:main_alg}

Let $B_n$ be the hyperoctahedral group. Let $\f = (f_1, \dots, f_s)$
and $\phi$  be polynomials in $\K[x_1, \dots, x_n]^{B_n}$ of degree at
most $d$.  In this section,  we first give some properties of polynomials
invariant by the action of $B_n$ and then we present our algorithm called {\sf
  Critical\_Hyperoctahedral} and its complexity to compute a
hyperoctahedral representation for $W(\phi, \f)$, assuming that this set is finite. 

\subsection{Hyperoctahedral polynomials}
\label{subsec:Hyper_polys}

%For a positive integer $n$, we denote by 
Let $\X$ denote the set of $n$ variables
$(x_1, \dots, x_n)$. For a polynomial $p$ in $\K[\X]^{B_n}$, we let
$g$ be a polynomial in $\K[z_1, \dots, z_n]$, with $(z_1, \dots, z_n)$
being new variables, such that $g(x_1^2, \dots, x_n^2) = p(x_1, \dots,
x_n)$. By the chain rule, we have for $1 \le i \le n, $ 
\begin{equation} \label{eq:chain_rule_remove}
\frac{\partial p}{\partial x_i} = \frac{\partial g}{\partial
  z_i}(x_1^2, \dots, x_n^2)\cdot  \frac{\partial x_i^2}{\partial x_i}
= 2x_i \cdot \frac{\partial g}{\partial z_i}(x_1^2, \dots, x_n^2).
\end{equation} 
The following lemma is a direct consequence
of~\eqref{eq:chain_rule_remove}.  

\begin{lemma}\label{lemma:some_zero_remove}
  Let $p$ be a polynomial in $\K[\X]^{B_n}$. Then for $1 \le i
  \le n$, the partial derivative $\frac{\partial p}{\partial x_i}$ of
  $p$ with respect to $x_i$ is divided by $x_i$. Consequently, the
  evaluation of $\frac{\partial p}{\partial x_i}$ at $x_i=0$ equals 
  zero. 
\end{lemma}
  \begin{example} Let us consider $n=5$ and a $B_5$-invariant
    polynomial $p =  \sum_{1 \le i < j \le 5}x_i^4x_j^4$. Then
    $\frac{\partial p}{\partial x_1}= 4x_1^3 \cdot (x_2^4+x_3^4+x_4^4
    + x_5^4)$, from which we can deduce that $\frac{\partial
      p}{\partial x_1}$ is divided by $x_1$. \label{ex:n5m4} 
  \end{example}

Let us denote by  $\frac{\partial p}{\partial \X} =
\big(\frac{\partial p}{\partial x_1}, \dots, \frac{\partial
  p}{\partial x_n} \big)$.  Let $m$ be a positive integer such that $m
\le n$ and assume further that $x_{m+1} = \cdots = x_n = 0$ with all other
$x_i$'s being nonzero. 
Then \[
\frac{\partial p}{\partial x_i}(x_1, \dots, x_m, 0, \dots, 0) =
\frac{\partial p(x_1, \dots, x_m, 0, \dots, 0)}{\partial x_i} .  
\] We consider the $\K$-algebra homomorphism
\begin{equation} \label{eq:homo} 
\begin{split}\mathbb{T}_m : \K[x_1, \dots, x_n] &\rightarrow  \K[x_1, \dots,
                                     x_m] \\ 
h(x_1, \dots, x_n) & \mapsto h(x_1, \dots, x_m, 0, \dots, 0).
\end{split}
\end{equation}
The differential operator and the homomorphism
$\mathbb{T}_m$ commute. In particular, for $1 \le i \le m$, we
have $\mathbb{T}_m\big(\frac{\partial p}{\partial x_i}\big) =
\frac{\partial p(\mathbb{T}_m(\X))}{\partial x_i}$.  Together
with Lemma~\ref{lemma:some_zero_remove}, this implies that 
\[
\mathbb{T}_m\Big( \frac{\partial p}{\partial \X} \Big) =
\Big(\frac{\partial p(\mathbb{T}_m(\X))}{\partial x_1}, \dots,
\frac{\partial p(\mathbb{T}_m(\X))}{\partial x_m}, 0, \dots, 0\Big). 
\]
\begin{example} \label{ex:mathbb}
Let $n=5$ and $p$ be the polynomial defined in Example~$\ref{ex:n5m4}$ and
$m=3$. Then $\mathbb{T}_3(p) = p(\mathbb{T}_3(\X)) = \sum_{1\le i , j
  \le 3}x_i^4x_j^4$ and  
\begin{multline*}
\Big(\frac{\partial p(\mathbb{T}_3(\X))}{\partial x_1}, \frac{\partial
  p(\mathbb{T}_3(\X))}{\partial x_2}, \frac{\partial
  p(\mathbb{T}_3(\X))}{\partial x_3}, 0, 0\Big)  \\ =
\big(4x_1^3(x_2^4+x_3^4), 4x_2^3(x_2^4+x_3^4), 4x_3^3(x_1^4+x_2^4), 0,
0\big) 
\end{multline*}
which is equal to $\mathbb{T}_3\Big( \frac{\partial p}{\partial \X}
\Big)$. 
\end{example}

Let $q$ be a polynomial in $\K[z_1, \dots, z_m]$ such that $$q(x_1^2,
\dots, x_m^2) = \mathbb{T}_m(p).$$  By the chain rule, we have $\frac{\partial
  \, \mathbb{T}_m(p)}{\partial x_i} = 2x_i \cdot \frac{\partial
  q}{\partial z_i}(x_1^2, \dots, x_m^2)$ for $i=1, \dots,
m$. Therefore, with $\Z = (z_1, \dots, z_m)$, we have
$\mathbb{T}_m\Big( \frac{\partial p}{\partial x_i}\Big) = 0$ for 
$i=m+1, \dots, n$ and
\[
\mathbb{T}_m\Big( \frac{\partial p}{\partial x_1}, \dots,
\frac{\partial p}{\partial x_m} \Big) = 2\cdot \frac{\partial q}{\partial
  \Z}(x_1^2, \dots, x_m^2) \cdot
\diag(x_1, \dots, x_m),
\] where $\diag(x_1, \dots, x_m)$ is the diagonal matrix with
diagonal entries $x_1, \dots, x_m$.

In general, let $\p = (p_1, \dots, p_\ell)$ be a sequence of
polynomials in $\K[x_1, \dots, x_n]^{B_n}$. We let $\q = (q_1, \dots,
q_\ell)$ be a sequence of polynomials in $\K[z_1, \dots, z_m]$ such
that  $q_i(x_1^2, \dots, x_m^2) = \mathbb{T}_m(p_i)$ for all $i=1,
\dots \ell$. Then the columns of $\mathbb{T}_m\big(\jac_{\X}(\p)\big)$
indexed by $m+1, \dots, n$ are zero columns and the first $m$ columns
are  
\[
 2 \cdot \jac_\Z(\q) (x_1^2, \dots, x_m^2) \cdot
\diag(x_1, \dots, x_m).
\]

Thus, if $(x_1, \dots, x_m)$ are nonzero,  then
$\mathbb{T}_m\big(\jac_{\X}(\p)\big)$ is rank deficient if and only if
$\jac_\Z(\q)$ is rank deficient. In addition, since $\p$ is
$B_n$-invariant, the sequence of polynomials $\q = \mathbb{T}_m(\p)$
is $S_m$-invariant. Note also that the set of points at which
$\jac_\Z(\q)$ is rank deficient may have zero coordinates.

\subsection{The main algorithm}
As we have seen in the introduction, if all variables $(x_1, \dots,
x_n)$ are nonzero, then we reduce to a new problem of computing critical
points set defined by symmetric polynomials. This new problem can be
solved by using algorithms in [13]. However, the set
$W(\phi, \f)$ can contain some points with zero coordinates.  

Let $m$ be a positive integer such that $m\le n$ and assume that $x_{m+1}
= \cdots = x_{n} = 0$ with all other $x_i$'s being nonzero. From
Subsection~\ref{subsec:Hyper_polys}, the matrix 
obtained from $\jac_\X(\f, \phi)$ by setting $x_{m+1} = \cdots = x_{n}
= 0$ contains $n-m$ zero columns. Let us denote by $\J_{[m]}$ a matrix in
$\K[x_1, \dots, x_m]^{(s+1)\times m}$ obtained by removing these zero
columns.  

Let $\q_{[m]} = (q_{[m], i}, \dots, q_{[m], s})$ and $\varphi_{[m]}$
be polynomials in $\K[z_1,$ $\dots, z_m]$ such that $$q_{[m], i}(x_1^2,
\dots, x_m^2) = \mathbb{T}_m(f_i) \,  {\rm for~all}  \, i=1, \dots,
s$$ and $ \varphi_{[m]}(x_1^2, \dots, x_m^2) = \phi$, where
$\mathbb{T}_m$ is the $\K$-algebra homomorphism defined
in~\eqref{eq:homo}. Then  from Subsection~\ref{subsec:Hyper_polys},
\[  
\J_{[m]} = 2 \cdot \jac_\Z(\q_{[m]}) (x_1^2, \dots, x_m^2) \cdot
\diag(x_1, \dots, x_m).
\]
As a consequence, if $x_i \ne 0$ for all $i=1, \dots, m$, it is
reduced to a new problem of computing the set $\mathsf{A}_{[m]} :=
W(\varphi_{[m]}, \q_{[m]}) \subset \Kbar{}^m$, with $\q_{[m]}$ and
$\varphi_{[m]}$ being invariant under the action of the symmetric group
$S_m$. In addition, since the degrees of $\f$ and $\phi$ are at most $d$
(an even positive integer), those of $\q_{[m]}$ and $\phi_{[m]}$ are
at most $d/2$.  
\begin{example} \label{ex:import_3}
Let us continue Example \ref{ex:import}. If $m=s=1$, i.e., $x_2 = x_3 = 0$  
 and $x_1 \ne 0$, then all the maximal minors of $\jac(f, \phi)$
 vanish. In this case, $q_{[1]} =  z_1^2-18$ and  $\phi_{[1]}
 = -3z_1$. If $m=2$, i.e.,  $x_3 = 0$, then $q_{[2]} =  z_1^2+z_2^2 -18$ and 
$\phi_{[2]} = -3z_1- 3z_2$. Finally, if  $m=3$, i.e., all $x_1, x_2, x_3$ are non-zero, then $q_{[3]} =
 z_1^2+z_2^2+z_3^2 - 18$ and $\phi_{[3]} =
 z_1z_2z_3-3z_1-3z_2-3z_3$.
\end{example}

We  claim that the algebraic set $W(\phi, \f)$ containing all critical
points of $\phi$ restricted to $V(\f)$ is invariant under the action
of $B_n$. 
\begin{lemma}
The algebraic set $W(\phi, \f)$ is $B_n$-invariant.
\end{lemma}
\begin{proof}
Let $E_m^*$ be the mapping defined in~\eqref{eq:E*}. Then the
algebraic set $W(\phi, \f)$ is equal to  
\[
\wbigcup_{m=1}^n \Big\{ \sigma(\a_{[m]},  \underbrace{0, \dots,
  0}_{n-m}) \, : \, \a_{[m]} \in E^*_m({\sf A}_{[m]})\, {\rm and~}
\sigma \in S_n \Big \}, 
\] which is $B_n$-invariant. Note that this union is not disjoint.  
\end{proof}

To compute a representation for the set $W(\varphi_{[m]}, \q_{[m]})$,
one can use the algorithm ${\sf Symmetric\_Representation}$ in
Theorem~\ref{them:FLSSV2021}. The output of the call ${\sf
  Symmetric\_Representation}(\q_{[m]},\varphi_{[m]})$ is a list of the
form $(\lambda_i, \scrR_i)_{1 \le i \le L}$, where $\lambda_i$ is a
partition of $m$ and $\scrR_i$ is a zero-dimensional parametrization
of a compression of the $S_m$-orbit of points in $\mathsf{A}_{[m]}$ of
pairwise distinct coordinates. Since the coordinates of points
represented by $\scrR_i$ are pairwise distinct but can equal to
zeros, we need to perform the function ${\sf Decompose\_Extended}$
from Lemma~\ref{lemma:decom_ext} to find the right types of those
points.  
\begin{example} \label{ex:import_4} Let us continue Example \ref{ex:import_3}.
\begin{itemize}
    \item  When $m=s=1$,  $q_{[1]} =  z_1^2-18$ and  $\phi_{[1]}
 = -3z_1$ and all the maximal minors of $\jac(f, \phi)$ vanish. Then the critical points set $W(\phi_{[1]}, q_{[1]})$ is empty.
 
 \item For $m=2$, we denote by $A := W(\phi_{[2]},
q_{[2]})$.  We consider the partition $(1^2)$ of $m=2$. Then the set
$W_{(1^2), [1]}$ of orbits of $W(\phi, f)$, parameterized  by $(1^2)$
and the integer $n-m = 1$, contains orbits of points of the form
$(e_1, e_2, 0)$ with $e_1^2 \ne e_2^2$ and $e_1, e_2 \ne 0$. The set
$W_{{\rm rem}, (1^2), [1]}$ of orbits of $W_{{\rm rem}}$ contains
orbits of points of the form $(a_1, a_2, 0)$ with $a_1 = e_1^2$ and
$a_2 = e_2^2$ for some $(e_1, e_2, 0) \in W_{(1^2), [1]}$. 
 The set $A_{(1^2)}$ of orbits of $A$, parameterized by $ (1^2)$, contains
 orbits of points of the form $(a_1, a_2)$ with $a_1 \ne a_2$. This set is defined by $q_{[2]}$, the determinant of
 $\jac(q_{[2]}, \phi_{[2]})$, and the open set $z_1 \ne z_2$, with
 $\det(\jac(q_{[2]}, \phi_{[2]})) = -6(z_1-z_2)$. Moreover, the
 solution set of the system $ z_1^2+z_2^2 -18 =  -6(z_1-z_2) = 0$ is
 empty in the open set $z_1 \ne  z_2$.  This means the set
 $A_{(1^2)}$ is  empty, which implies that the set $W_{(1^2), [1]}$ is
 empty as well. 
 
 \hspace{0.3 cm} For the partition $(2^1)$ of $m=2$,
   the set $W_{(2^1), [1]}$ of orbits of $W(\phi,
   f)$, parameterized  by $(2^1)$ and the integer $n-m = 1$, contains orbits 
   of points of the form $(e_1, e_2, 0)$ with $e_1^2 = e_2^2$ and $e_1
   \ne 0$.  The set $W_{{\rm rem}, (2^1), [1]}$ of orbits of $W_{{\rm
       rem}}$ is orbits of   points of the form $(a_1, a_1, 0)$ with
   $a_1 = e_1^2 = e_2^2$ for some $(e_1, e_2 , 0) \in W_{(1^2),
     [1]}$.
   The set $A_{(2^1)}$ of orbits of $A$, parameterized by $ (2^1)$, is
   orbits  of points of the form $(a_1, a_1)$.  Since all
   the maximal minors of $\jac(q_{[2]}, \phi_{[2]})$ are zero in the
   closed set $V(z_1-z_2)$, the set $A_{(2^1)}$  is defined by
   $q_{[2]}(z_1, z_1) :=  2z_1^2-18$ and $z_1-z_2$. A zero-dimensional
   parametrization of $V(q_{[2]})$ is $\scrR_1 = ((z^2-9, 18), z_1)$
   with the last one  being a linear form  in $z_1$. The  procedure
   ${\sf Decompose\_Extended}((2^1),$ $\scrR_1 ,1)$ outputs $((2^1),
   \scrR_1 ,{\bf 0}_1)$.

    \hspace{0.3 cm} Thus,  $\scrL_1 = (\gamma_1, \scrR_1, {\bf 0}_1)$. Here
   $\gamma_1 = (2^1)$ is a partition of $2$ and  $\scrR_1 =((t^2-9, 18),
   z_1)$ is a zero-dimensional parametrization of $L_{\gamma_1}(W_{\rm
   rem} \cap U_{\gamma_1, [1], {\rm dist}})$.

 \item For $m=3$, we denote by $B := W(\phi_{[3]},
 q_{[3]})$. Consider the partition $(1^1\, 2^1)$ of $m=3$. 
  The set $W_{(1^1\, 2^1), [0]}$ of $B_3$-orbits of
  $W(\phi, f)$, parameterized  by $(1^1\, 2^1)$ and the integer $n-m = 
  0$, contains orbits of points of the form $(e_1, e_2, e_3)$ with $e_1^2
  \ne e_2^2$, $e_2^2 = e_3^2 $, and  $e_1, e_2, e_3 \ne 0$.   The set
  $W_{{\rm rem}, (1^1\, 2^1), [0]}$ of $S_3$-orbits of $W_{{\rm rem}}$
  contains orbits of points of the form $(a_1, a_2, a_2)$ with $a_1 =
  e_1^2$, $a_2 = e_2^2$ for some $(e_1, e_2 , e_3) \in W_{(1^1\, 2^1),
    [0]}$.  

 \hspace{0.3 cm} The set $B_{(1^1\, 2^1)}$ of $S_3$-orbits of $B$ parameterized by $
 (1^1\, 2^1)$ is orbits of points of the form $(a_1, a_2, a_2)$ with
 $a_1 \ne a_2$. It contains the orbits of points in the zero set
 of \begin{equation} \label{eq:first} (z_1^2+2z_2^2 -18,
   2(z_1z_2+z_2^2-3)(z_1-z_2), z_2-z_3).  
 \end{equation}
 The first polynomial is equal to $q_{[3]}(z_1, z_2,
 z_2)$ and second polynomial equals the evaluation of  a non-zero $2$-minor of
 $\jac_{z_1, z_2, z_3}(q_{[3]}, \phi_{[3]})$ at $z_3 = z_2$.   

  \hspace{0.3 cm}  A zero-dimensional parametrization of the zero set of equations in \eqref{eq:first}  is $\scrR = ((v, v_1, v_2), z_1) =
 (((t^2-6)(t^4-8t^2+3), 384t^4-2748t^2+1260, 28t^4-204t^2+108), z_1)$, 
 where the last one is  a random linear form in $(z_1, z_2)$. We then
 perform the  procedure ${\sf Decompose\_Extended}((1^1\,2^1),$ $\scrR
 ,{\bf 0}_0)$. The output of this procedure is $((\gamma_2,
 \scrR_2,{\bf 0}_0), (\gamma_3, \scrR_3,{\bf 0}_0))$, with
 \[
 \gamma_2 = (3^1), \scrR_2 = ((t^2-6, 12), z_1)\] and $
  \gamma_3 = (1^1\, 2^1), \scrR_3 = ((t^4-8t^2+3, -4t^2-36, 16t^2-12),
 z_1)$. 
  This implies that the zero set of \eqref{eq:first} contains
 $B_n$-orbits of two types $((3^1), [0])$ and $((1^1, 2^1),
 [0])$. Here $\scrR_2$ and  $\scrR_3$ are 
 zero-dimensional parametrizations  of the sets $L_{\gamma_2}(W_{\rm
   rem} \cap U_{\gamma_2, [0], {\rm dist}})$ and $L_{\gamma_3}(W_{\rm
   rem}$ $\cap U_{\gamma_3, [0], {\rm dist}})$ respectively.

\smallskip
\hspace{0.3 cm}
For the partition $(1^3)$ of $m=3$, the set $W_{(1^3), [0]}$ of $B_3$-orbits of
  $W(\phi, f)$, parameterized  by $(1^3)$ and the integer $n-m = 
  0$, contains orbits of points of the form $(e_1, e_2, e_3)$ with $e_i^2
  \ne e_j^2$ for $1 \le i < j \le 3$ and $e_i \ne 0$ for all $i=1, \dots,
  3$.   The set $W_{{\rm rem}, (1^3), [0]}$ of $S_3$-orbits of
  $W_{{\rm rem}}$ is orbits of points of the form $(a_1, a_2, a_3)$
  with $a_i = e_i^2$ for some $(e_1, e_2 , e_3)   \in W_{(1^3), [0]}$.  
 
\hspace{0.3 cm}
  The set $B_{(1^3)}$ of $S_3$-orbits of $B$ parameterized by $
   (1^3)$ is orbits of points of the form $(a_1, a_2,  a_3)$
   with $a_i \ne a_j$ for $1 \le i < j \le 3$. By using results from [13], the set  $B_{(1^3)}$ is equal to the zero set of 
   $
   q_{[3]}, -6, 2(x_1+x_2+x_3), -2, 
   $ which is empty.  

\smallskip
\hspace{0.3 cm} For {the partition $(3^1)$ of $m=3$}, the set $W_{(3^1), [0]}$ of $B_3$-orbits of $W(\phi,
 f)$, parameterized  by $(3^1)$ and the integer $n-m = 0$, contains orbits
 of points of the form $(e_1, e_2, e_3)$ with $e_i^2 = e_j^2$ for $1
 \le i < j \le 3$. The set 
  $W_{{\rm rem}, (3^1), [0]}$ of $S_3$-orbits of $W_{{\rm rem}}$ is
  orbits of points of the form $(a_1, a_1, a_1)$   with $a_1 = e_i^2$
  for $i=1, \dots 3$, for some $(e_1, e_2 , e_3)   \in W_{(3^1),
    [0]}$.

 \hspace{0.3 cm}  The set $B_{(3^1)}$ of $S_3$-orbits of $B$, parameterized by $
   (3^1)$, contains orbits of points of the form $(a_1, a_1,  a_1)$. This
   means we consider $z_1=z_2=z_3$; so all the maximal minors of
   $\jac(q_{[3]}, \phi_{[3]})$ vanish. In addition, $q_{[3]}(z_1, z_1,
   z_1) = 3z_1^2  -18$ and $\phi_{[3]}(z_1, z_1, z_1) =
   z_1^3-9z_1$.  
   
    \hspace{0.3 cm}A zero-dimensional parametrization of $V(q_{[3]}(z_1, z_1,
   z_1))$ is $\scrR_4 =((t^2-6, 12), z_1)$ with the last one  being
   a linear form in $z_1$. Finally, the procedure ${\sf
     Decompose\_Extended}((3^1), \scrR_4 , $ $0)$ outputs $((3^1),
   \scrR_4,{\bf 0}_0)$. This implies that the type  of any points in
   $V(q_{[3]}(z_1, z_1, z_1))$ is $(\gamma_4, [0])$ with $\gamma_4  =  (3^1)$. 

 \hspace{0.3 cm} Thus, $\scrL_2 = ((3^1), \scrR_2, {\bf 0}_0),
\scrL_3 = ((1^2 \, 2^1),\scrR_3, {\bf 0}_0)$, and $\scrL_4 = ((3^1), \scrR_4, {\bf 0}_0))$ with $\scrR_2 = \scrR_4 =((t^2-6, 12), z_1)$ and $\scrR_3 = ((t^4-8t^2+3, -4t^2-36, 16t^2-12),
 z_1)$. 
\end{itemize}
\end{example}
Note that when $m=s$, all maximal minors of $\jac_{\X}(\f, \phi)$
vanish at $x_{m+1} = \dots = x_n = 0$. Then $W(\varphi_{[m]},
\q_{[m]}) = V(\q_{[m]})$. Moreover, performing Steps \ref{step:qphi} to
\ref{step:all} gives us representations of the set  
of orbits of $W(\phi, \f)_{{\rm rem}} \cap U_{\lambda, [n-m]}$ for
partitions  $\lambda$ of $m$. This set contains all orbits
of points in $W(\phi, \f)_{{\rm rem}}$  whose type $\lambda'$ with
$\lambda' >_{\rm  ext}\lambda$. Therefore, it suffices to do
computations for $m$ at least $s$.

Apart from the above subroutines, we also need a procedure ${\sf
  Remove\_Duplicates}(\scrL)$. Here $\scrL = (\lambda_i, \scrR_i, {\bf 0}_k)_{1
  \le i \le M}$ is a list, where $\lambda_i$ is a partition of length $k$ of $m$,
with $s \le m \le n$, and $\scrR_i$ is a zero-dimensional
parametrization. Since $\scrL$ may not contain distinct $\lambda_i$'s,
the procedure ${\sf Remove\_Duplicates}$  removes $(\lambda_i,
\scrR_i, {\bf 0}_k)$ from $\scrL$ so that all resulting partitions are pairwise
distinct. We can choose any entries among the repetitions to remove
from the list. 

\begin{example}
Let us continue Example \ref{ex:import_4}. For the input of $f$ and $\phi$ in our
   example, we obtain $\scrL = ((\gamma_i, \scrR_i,  {\bf 0}_k))_{i=1, \dots, 4}$ with $\gamma_2 = \gamma_4$ and $\scrR_2 = \scrR_4$. Therefore, we need to use the procedure ${\sf 
    Remove\_Duplicates}(\scrL)$  to remove one of $(\gamma_2,
  \scrR_2, {\bf 0}_0)$ and $(\gamma_4, \scrR_4,$ ${\bf 0}_0)$, for
  example $(\gamma_4, \scrR_4, {\bf 0}_0)$. Hence, the output of our
  algorithm is $$((\gamma_1, \scrR_1, {\bf 0}_1), (\gamma_2,\scrR_2,
  {\bf 0}_0), (\gamma_3, \scrR_3, {\bf 0}_0)).$$  The sizes of
  $(Z(\scrR_i))_{1 \le i \le 3}$ are  respectively $2, 2$, and $4$. Therefore, the
  total number of points computed by our algorithm is $8$, which is
  significantly smaller than the size 148 of the set $W(\phi, f)$, as claimed in Example \ref{ex:import_2}. 

\end{example}
\begin{algorithm}[H] 	 
  \caption{${\sf Critical\_Hyperoctahedral}$}
   {\bf Input:} A sequence $\f = (f_1, \dots, f_s)$ and
  a  map $\phi$ in $\K[x_1, \dots, x_n]^{B_n}$\\ with $W(\phi, \f)
  \subset \Kbar{}^n$ is finite. \\
 \hspace{-1.4cm}{\bf Output:} A hyperoctahedral representation  for
  $W(\phi, \f)$. 

 \begin{enumerate}
\item $\scrL = [\, ]$
\item for $m$ from $s$ to $n$ do: \label{Step:for}
\begin{enumerate}
\item find $\q_{[m]} = \mathbb{T}_m(\f)$ and $\varphi_{[m]} =
  \mathbb{T}_m(\phi)$ \label{step:qphi} 
\item \label{step:SYM_REP}compute $$\quad \quad (\lambda_i,
  \scrR_i)_{1 \le i \le L} = {\sf Symmetric\_Representation}(\q_{[m]},
  \varphi_{[m]})$$  
\item for $i$ from $1$ to $L$ do 
\begin{enumerate} \label{step:all}
    \item compute $(\lambda_\ell, \scrR_{\lambda_\ell}, {\bf 0}_{n-m'})_{1 \le \ell \le L'}$ \newline $ ={\sf
        Decompose\_Extended}(\lambda_i, \mathscr{R}_{\lambda_i}, n-t)$, where $t$ and $m'$ are integers such that $\lambda_i$ and $\lambda_\ell$ are respectively partitions of $t$ and $m'$
        \item append $(\lambda_\ell, \scrR_{\lambda_\ell}, {\bf 0}_{n-m'})_{1 \le \ell \le L'}$ to $\scrL$%, where
          %$\ell_0$ is such that $\lambda_{\ell_0}$ is a partition of $m$, if
          %such $\ell_0$ exists 
\end{enumerate}
\end{enumerate}
\item return ${\sf Remove\_Duplicates}(\scrL)$ 
\end{enumerate}
  \label{alg:some_zeros} 	 
\end{algorithm}  

\begin{theorem}
 Algorithm~\ref{alg:some_zeros} is correct and it uses $$\softO\big(
 (d/2)^2 C(C^5+E)n^{10} \Gamma\big)$$ operations in $\K$, where \[ 
C = (d/2)^s {{n+d/2-1}\choose{n}}, E = n(d/2+1){{n+d/2} \choose{n}},\]
and\[ \Gamma = n^2{{n+d/2}\choose {d/2}} + n^4 {{n}\choose {s+1}}.  
\] Moreover, the total number of points described by the output is at
most $n\,C$ \label{thm:main_thm}
\end{theorem} 

\begin{proof}
The correctness of the algorithm holds from the previous discussions. It remains
to establish a complexity analysis of our algorithm. 

Applying $\mathbb{T}_m$ to $\f$ and $\phi$
takes linear time in the total number of monomials in $\f,
\phi$. Since $\f$ and $\phi$ are $B_n$-invariant, all monomials in
$\f$ and $\phi$ have even degrees at most $d$. The total number
of monomials in $\f, \phi$ is then $O\big((s+1) {{n+d/2}\choose
  {d/2}}\big)$, with the same cost is needed to find $\q_{[m]}$ and
$\varphi_{[m]}$ at Step~\ref{step:qphi}. Recall that the degrees of
$\q_{[m]}$ and $\varphi_{[m]}$ are bounded by $d/2$. 
 
Theorem~\ref{them:FLSSV2021} implies that computing $(\lambda_i,
\scrR_i)_{1 \le i \le L}$ at Step~\ref{step:SYM_REP} requires  
$
\softO\big((d/2)^2c_m(e_m+c_m^5)m^9 \rho_m \big)
$ operations in $\K$, where $c_m = (d/2)^s {{m+d/2-1}\choose{m}}$,
$e = m(d/2+1){{m+d/2} \choose{m}}$, and $\rho_m = m^2{{m+d/2}\choose
  {d/2}} + m^4 {{m}\choose {s+1}}$. Moreover the number of solutions
in the output of 
the call ${\sf Symmetric\_Representation}$ is at most $c_m$ by using 
Theorem~\ref{them:FLSSV2021}. The
cost of the procedure ${\sf  Decompose\_Extended}$ at Step~\ref{step:all} is then
$\softO{(c_m^2m)}$ operations in $\K$, which is negligible in
comparison to the other costs. 

Thus,  the total cost for running
Step~\ref{Step:for} is  
        \[
\softO\big(\sum_{m=s}^n(d/2)^2c_m(c_m^5+c_m^5)m^9 \rho_m \big) =
\softO\big( (d/2)^2 C(C^5+E)n^{10} \Gamma\big)
\] operations in $\K$. Since the total size of $(\scrR_k)_{1 \le k \le
  L}$ is at most $c_m$ which is bounded by $C$, the the size of the
input for the procedure ${\sf Remove\_Duplicates}(\scrL)$ is at most
$(n-s)C \le nC$. This means the cost of the call ${\sf
  Remove\_Duplicates}(\scrL)$ is negligible.  
\end{proof}

%%%%%%%%%%%%%%%%%%%%%%%%%%%%%%%
%%%%%%%%%%%%%%%%%%%%%%%%%%%%%%%
%%%%%%%%%%%%%%%%%%%%%%%%%%%%%%%

\section{Experimental runs}
\label{sec:expruns}
\begin{small}
\begin{center}
\begin{table}
  \centering
  \begin{tabular}{l |l | l | l | l | l | l}
 $(n,s)$ & $d$& Time({\sf H}) &  Time$({\sf N })$ & Size({\sf H})  &
                                                                    Size({\sf
                                                                    N})
    & $n\,C$  \\ \hline 
 (3, 1) & 8 & ~~ 0.775 &~~ 0.177&~~45&~~187 & 240 \\
(3, 2)& 8 &  ~~ 1.016 &~~ 0.124&~~ 38&~~187& 960 \\
\hline 
(4, 1)& 8 &  ~~ 2.140 &~~ 16.948&~~ 50&~~455& 560     \\ 
(4, 2)& 8 &  ~~ 2.968 &~~ 66.344&~~ 66&~~627& 2240 \\
(4, 3)& 8 &  ~~ 3.008 &~~ 12.268&~~ 40&~~373& 8960 \\
\hline 
(5, 1)& 12 &  ~~ 35.373 &\hspace{0.3cm}-&~~ 482&\hspace{0.3cm}-& 7560 \\
(5, 2)& 12 &  ~~ 61.912 &\hspace{0.3cm}-&~~ 956&\hspace{0.3cm}-&45360 \\
(5, 3)& 12 &  ~~ 81.432 &\hspace{0.3cm}-&~~ 848&\hspace{0.3cm}-& 272160 \\
(5, 4)& 12 &  ~~  72.681&\hspace{0.3cm}-&~~ 352&\hspace{0.3cm}-&
                                                                 1632960 \\
\hline
(6, 1)& 12 &  ~~ 106.444 &\hspace{0.3cm}-&~~ 493&\hspace{0.3cm}-&
                                                                  16632 \\
(6, 2)& 12 &  ~~ 160.596 &\hspace{0.3cm}-&~~ 1086&\hspace{0.3cm}-&
                                                                   99792 \\
(6, 3)& 12 &  ~~ 208.375 &\hspace{0.3cm}- &~~ 1289&\hspace{0.3cm}-&
                                                                    598752 \\
(6, 4)& 12 &  ~~ 203.349 &\hspace{0.3cm}-&~~ 920&\hspace{0.3cm}-&
                                                                  3592512 \\
(6, 5)& 12 &  ~~ 160.681 &\hspace{0.3cm}-&~~ 360&\hspace{0.3cm}-&  21555072
  \end{tabular}
\vspace{0.2cm}
  \caption{Some experiments support the main algorithm}
  \label{tab:forpol}
\end{table}
\end{center}
\end{small}
In this section we report on some experimental runs supporting the results in this
paper. We will compare our {${\sf Critical\_Hyperoctahedral}$} algorithm 
with a naive algorithm.  A latteralgorithm computes a zero-dimensional
parametrization of $W(\f, \phi)$ without exploiting any  invariance
properties of the inputs. In this case it runs  Gr\"{o}bner bases computations to
solve the polynomial system consisting of $\f$ and all the
$(s+1)$-minors of $\jac_\X(\f, \phi)$.   

We use  the Maple computer algebra system running on a computer with
16 GB RAM to do our experiments. The Gr\"{o}bner bases computations in
Maple  use the implementation of the F4 and FGLM algorithms from the
FGb package [12].

 The $B_n$-invariant polynomials $\f$ and $\phi$ are chosen uniformly
 at random in $\K[x_1, \dots, x_n]$ with $\K = GF(65521)$, a finite
 field of $65521$ elements. The degrees of $f_1, \dots, f_s$ and
 $\phi$ are all equal to $d$, a positive even number.  

In Table~\ref{tab:forpol}, we denote by  Time({\sf H}) and Time$({\sf
  N })$  the total times (in second) spent in our algorithm and in the
naive one, respectively. The experiments, with the corresponding time marked with a dash, 
 are stopped once the computation has gone past 24
 hours. 

Finally, Size$({\sf N })$ and  Size({\sf H}) are the numbers of points 
computed by our algorithm and by the naive one. Size({\sf H}) is the
cardinality of the whole set $W(\phi, \f)$. We also compute the 
bound $n\,C$ for Size$({\sf N })$, where $C$ is defined in
Theorem~\ref{thm:main}. It can be seen that $n\,C$ is not a  tight
bound for the numbers of points computed by our algorithm. However, we
want to emphasize that the runtime of our algorithm is polynomial in the
sizes of both the output and the input.

\balance 
\bibliographystyle{ACM-Reference-Format}
\bibliography{sample-base}


\begin{thebibliography}{33}

%%% the Web address.

\ifx \showCODEN    \undefined \def \showCODEN     #1{\unskip}     \fi
\ifx \showDOI      \undefined \def \showDOI       #1{#1}\fi
\ifx \showISBNx    \undefined \def \showISBNx     #1{\unskip}     \fi
\ifx \showISBNxiii \undefined \def \showISBNxiii  #1{\unskip}     \fi
\ifx \showISSN     \undefined \def \showISSN      #1{\unskip}     \fi
\ifx \showLCCN     \undefined \def \showLCCN      #1{\unskip}     \fi
\ifx \shownote     \undefined \def \shownote      #1{#1}          \fi
\ifx \showarticletitle \undefined \def \showarticletitle #1{#1}   \fi
\ifx \showURL      \undefined \def \showURL       {\relax}        \fi
% invisible to TeX
\providecommand\bibfield[2]{#2}
\providecommand\bibinfo[2]{#2}
\providecommand\natexlab[1]{#1}
\providecommand\showeprint[2][]{arXiv:#2}

\bibitem[\protect\citeauthoryear{Alonso, Becker, Roy, and W{\"o}rmann}{Alonso
  et~al\mbox{.}}{1996}]%
        {alonso1996zeros}
\bibfield{author}{\bibinfo{person}{M.-E. Alonso}, \bibinfo{person}{E. Becker},
  \bibinfo{person}{M.-F. Roy}, {and} \bibinfo{person}{T. W{\"o}rmann}.}
  \bibinfo{year}{1996}\natexlab{}.
\newblock \showarticletitle{Zeros, multiplicities, and idempotents for
  zero-dimensional systems}.
\newblock In \bibinfo{booktitle}{\emph{Algorithms in algebraic geometry and
  applications}}. \bibinfo{publisher}{Springer},
  \bibinfo{address}{Birkh{\"a}user Basel}, \bibinfo{pages}{1--15}.
\newblock


\bibitem[\protect\citeauthoryear{Aubry, Rouillier, and Safey El~Din}{Aubry
  et~al\mbox{.}}{2002}]%
        {aubry2002real}
\bibfield{author}{\bibinfo{person}{P. Aubry}, \bibinfo{person}{F. Rouillier},
  {and} \bibinfo{person}{M. Safey El~Din}.} \bibinfo{year}{2002}\natexlab{}.
\newblock \showarticletitle{Real solving for positive dimensional systems}.
\newblock \bibinfo{journal}{\emph{Journal of Symbolic Computation}}
  \bibinfo{volume}{34}, \bibinfo{number}{6} (\bibinfo{year}{2002}),
  \bibinfo{pages}{543--560}.
\newblock


\bibitem[\protect\citeauthoryear{Bank, Giusti, Heintz, and Mbakop}{Bank
  et~al\mbox{.}}{2001}]%
        {bank2001polar}
\bibfield{author}{\bibinfo{person}{B. Bank}, \bibinfo{person}{M. Giusti},
  \bibinfo{person}{J. Heintz}, {and} \bibinfo{person}{G.-M. Mbakop}.}
  \bibinfo{year}{2001}\natexlab{}.
\newblock \showarticletitle{Polar varieties and efficient real elimination}.
\newblock \bibinfo{journal}{\emph{Mathematische Zeitschrift}}
  \bibinfo{volume}{238}, \bibinfo{number}{1} (\bibinfo{year}{2001}),
  \bibinfo{pages}{115--144}.
\newblock


\bibitem[\protect\citeauthoryear{Bank, Giusti, Heintz, and Pardo}{Bank
  et~al\mbox{.}}{2005}]%
        {bank2005generalized}
\bibfield{author}{\bibinfo{person}{B. Bank}, \bibinfo{person}{M. Giusti},
  \bibinfo{person}{J. Heintz}, {and} \bibinfo{person}{L.-M. Pardo}.}
  \bibinfo{year}{2005}\natexlab{}.
\newblock \showarticletitle{Generalized polar varieties: Geometry and
  algorithms}.
\newblock \bibinfo{journal}{\emph{Journal of complexity}} \bibinfo{volume}{21},
  \bibinfo{number}{4} (\bibinfo{year}{2005}), \bibinfo{pages}{377--412}.
\newblock


\bibitem[\protect\citeauthoryear{Bank, Giusti, Heintz, and Safey El~Din}{Bank
  et~al\mbox{.}}{2014}]%
        {bank2014intrinsic}
\bibfield{author}{\bibinfo{person}{B. Bank}, \bibinfo{person}{M. Giusti},
  \bibinfo{person}{J. Heintz}, {and} \bibinfo{person}{M. Safey El~Din}.}
  \bibinfo{year}{2014}\natexlab{}.
\newblock \showarticletitle{Intrinsic complexity estimates in polynomial
  optimization}.
\newblock \bibinfo{journal}{\emph{Journal of Complexity}} \bibinfo{volume}{30},
  \bibinfo{number}{4} (\bibinfo{year}{2014}), \bibinfo{pages}{430--443}.
\newblock


\bibitem[\protect\citeauthoryear{Bank, Giusti, Heintz, Safey El~Din, and
  Schost}{Bank et~al\mbox{.}}{2010}]%
        {bank2010geometry}
\bibfield{author}{\bibinfo{person}{B. Bank}, \bibinfo{person}{M. Giusti},
  \bibinfo{person}{J. Heintz}, \bibinfo{person}{M. Safey El~Din}, {and}
  \bibinfo{person}{\'E. Schost}.} \bibinfo{year}{2010}\natexlab{}.
\newblock \showarticletitle{On the geometry of polar varieties}.
\newblock \bibinfo{journal}{\emph{Applicable Algebra in Engineering,
  Communication and Computing}} \bibinfo{volume}{21}, \bibinfo{number}{1}
  (\bibinfo{year}{2010}), \bibinfo{pages}{33--83}.
\newblock


\bibitem[\protect\citeauthoryear{Basu, Pollack, and Roy}{Basu
  et~al\mbox{.}}{2006}]%
        {Basu2006}
\bibfield{author}{\bibinfo{person}{S. Basu}, \bibinfo{person}{R. Pollack},
  {and} \bibinfo{person}{M.-F. Roy}.} \bibinfo{year}{2006}\natexlab{}.
\newblock \bibinfo{booktitle}{\emph{Algorithms in Real Algebraic Geometry
  (Algorithms and Computation in Mathematics)}}.
\newblock \bibinfo{publisher}{Springer-Verlag}, \bibinfo{address}{Berlin,
  Heidelberg}.
\newblock
\showISBNx{3540330984}


\bibitem[\protect\citeauthoryear{Basu, Roy, Safey El~Din, and Schost}{Basu
  et~al\mbox{.}}{2014}]%
        {basu2014baby}
\bibfield{author}{\bibinfo{person}{S. Basu}, \bibinfo{person}{M.-F. Roy},
  \bibinfo{person}{M. Safey El~Din}, {and} \bibinfo{person}{{\'E}. Schost}.}
  \bibinfo{year}{2014}\natexlab{}.
\newblock \showarticletitle{A baby step--giant step roadmap algorithm for
  general algebraic sets}.
\newblock \bibinfo{journal}{\emph{Foundations of Computational Mathematics}}
  \bibinfo{volume}{14}, \bibinfo{number}{6} (\bibinfo{year}{2014}),
  \bibinfo{pages}{1117--1172}.
\newblock


\bibitem[\protect\citeauthoryear{Benson}{Benson}{1993}]%
        {benson1993}
\bibfield{author}{\bibinfo{person}{David~J Benson}.}
  \bibinfo{year}{1993}\natexlab{}.
\newblock \bibinfo{booktitle}{\emph{Polynomial invariants of finite groups}}.
  Vol.~\bibinfo{volume}{190}.
\newblock \bibinfo{publisher}{Cambridge University Press},
  \bibinfo{address}{New York, NY, USA}.
\newblock


\bibitem[\protect\citeauthoryear{Dahan and Schost}{Dahan and Schost}{2004}]%
        {dahan2004sharp}
\bibfield{author}{\bibinfo{person}{X. Dahan} {and} \bibinfo{person}{{\'E}.
  Schost}.} \bibinfo{year}{2004}\natexlab{}.
\newblock \showarticletitle{Sharp estimates for triangular sets}. In
  \bibinfo{booktitle}{\emph{Proceedings of the 2004 international symposium on
  Symbolic and algebraic computation}}. \bibinfo{publisher}{ACM},
  \bibinfo{address}{University of Cantabria, Santander, Spain},
  \bibinfo{pages}{103--110}.
\newblock


\bibitem[\protect\citeauthoryear{Eisenbud}{Eisenbud}{2013}]%
        {eisenbud2013commutative}
\bibfield{author}{\bibinfo{person}{D. Eisenbud}.}
  \bibinfo{year}{2013}\natexlab{}.
\newblock \bibinfo{booktitle}{\emph{Commutative algebra: with a view toward
  algebraic geometry}}. Vol.~\bibinfo{volume}{150}.
\newblock \bibinfo{publisher}{Springer Science \& Business Media},
  \bibinfo{address}{Berlin, Heidelberg}.
\newblock


\bibitem[\protect\citeauthoryear{Faug{\`e}re}{Faug{\`e}re}{2010}]%
        {faugere2010fgb}
\bibfield{author}{\bibinfo{person}{J.-C. Faug{\`e}re}.}
  \bibinfo{year}{2010}\natexlab{}.
\newblock \showarticletitle{FGb: a library for computing Gr{\"o}bner bases}. In
  \bibinfo{booktitle}{\emph{International Congress on Mathematical Software}}.
  \bibinfo{publisher}{Springer}, \bibinfo{address}{Berlin, Heidelberg},
  \bibinfo{pages}{84--87}.
\newblock


\bibitem[\protect\citeauthoryear{Faug\`ere, Labahn, Safey El~Din, Schost, and
  Vu}{Faug\`ere et~al\mbox{.}}{2020}]%
        {FLSSV2021}
\bibfield{author}{\bibinfo{person}{J-C. Faug\`ere}, \bibinfo{person}{G.
  Labahn}, \bibinfo{person}{M. Safey El~Din}, \bibinfo{person}{{\'{E}}.
  Schost}, {and} \bibinfo{person}{T.~X. Vu}.} \bibinfo{year}{2020}\natexlab{}.
\newblock \showarticletitle{Computing critical points for invariant algebraic
  systems}.
\newblock \bibinfo{journal}{\emph{arXiv preprint arXiv:2009.00847}}
  (\bibinfo{year}{2020}).
\newblock


\bibitem[\protect\citeauthoryear{Faug{\`e}re, Safey El~Din, and
  Spaenlehauer}{Faug{\`e}re et~al\mbox{.}}{2012}]%
        {faugere2012critical}
\bibfield{author}{\bibinfo{person}{J-C. Faug{\`e}re}, \bibinfo{person}{M. Safey
  El~Din}, {and} \bibinfo{person}{P-J. Spaenlehauer}.}
  \bibinfo{year}{2012}\natexlab{}.
\newblock \showarticletitle{Critical points and Gr{\"o}bner bases: the unmixed
  case}. In \bibinfo{booktitle}{\emph{Proceedings of the 37th International
  Symposium on Symbolic and Algebraic Computation}}. \bibinfo{publisher}{ACM},
  \bibinfo{address}{Grenoble, France}, \bibinfo{pages}{162--169}.
\newblock


\bibitem[\protect\citeauthoryear{Gathen and Gerhard}{Gathen and
  Gerhard}{2003}]%
        {Gat03}
\bibfield{author}{\bibinfo{person}{J.~V.~Z. Gathen} {and} \bibinfo{person}{J.
  Gerhard}.} \bibinfo{year}{2003}\natexlab{}.
\newblock \bibinfo{booktitle}{\emph{Modern Computer Algebra}
  (\bibinfo{edition}{2} ed.)}.
\newblock \bibinfo{publisher}{Cambridge University Press},
  \bibinfo{address}{New York, NY, USA}.
\newblock
\showISBNx{0521826462}


\bibitem[\protect\citeauthoryear{Gianni and Mora}{Gianni and Mora}{1987}]%
        {gianni1987algebraic}
\bibfield{author}{\bibinfo{person}{P. Gianni} {and} \bibinfo{person}{T. Mora}.}
  \bibinfo{year}{1987}\natexlab{}.
\newblock \showarticletitle{Algebraic solution of systems of polynomial
  equations using Gr{\"o}bner bases}. In
  \bibinfo{booktitle}{\emph{International Conference on Applied Algebra,
  Algebraic Algorithms, and Error-Correcting Codes}}.
  \bibinfo{publisher}{Springer}, \bibinfo{address}{Berlin, Heidelberg},
  \bibinfo{pages}{247--257}.
\newblock


\bibitem[\protect\citeauthoryear{Giusti, Heintz, Morais, Morgenstem, and
  Pardo}{Giusti et~al\mbox{.}}{1998}]%
        {giusti1998straight}
\bibfield{author}{\bibinfo{person}{M. Giusti}, \bibinfo{person}{J. Heintz},
  \bibinfo{person}{J.~E. Morais}, \bibinfo{person}{J. Morgenstem}, {and}
  \bibinfo{person}{L.~M. Pardo}.} \bibinfo{year}{1998}\natexlab{}.
\newblock \showarticletitle{Straight-line programs in geometric elimination
  theory}.
\newblock \bibinfo{journal}{\emph{Journal of pure and applied algebra}}
  \bibinfo{volume}{124}, \bibinfo{number}{1-3} (\bibinfo{year}{1998}),
  \bibinfo{pages}{101--146}.
\newblock


\bibitem[\protect\citeauthoryear{Giusti, Heintz, Morais, and Pardo}{Giusti
  et~al\mbox{.}}{1995}]%
        {giusti1995polynomial}
\bibfield{author}{\bibinfo{person}{M. Giusti}, \bibinfo{person}{J. Heintz},
  \bibinfo{person}{J.~E. Morais}, {and} \bibinfo{person}{L.~M. Pardo}.}
  \bibinfo{year}{1995}\natexlab{}.
\newblock \showarticletitle{When polynomial equation systems can be solved
  fast?}. In \bibinfo{booktitle}{\emph{International Symposium on Applied
  Algebra, Algebraic Algorithms, and Error-Correcting Codes}}.
  \bibinfo{publisher}{Springer}, \bibinfo{address}{Berlin, Heidelberg},
  \bibinfo{pages}{205--231}.
\newblock


\bibitem[\protect\citeauthoryear{Giusti, Lecerf, and Salvy}{Giusti
  et~al\mbox{.}}{2001}]%
        {giusti2001grobner}
\bibfield{author}{\bibinfo{person}{M. Giusti}, \bibinfo{person}{G. Lecerf},
  {and} \bibinfo{person}{B. Salvy}.} \bibinfo{year}{2001}\natexlab{}.
\newblock \showarticletitle{A Gr{\"o}bner free alternative for polynomial
  system solving}.
\newblock \bibinfo{journal}{\emph{Journal of complexity}} \bibinfo{volume}{17},
  \bibinfo{number}{1} (\bibinfo{year}{2001}), \bibinfo{pages}{154--211}.
\newblock


\bibitem[\protect\citeauthoryear{Greuet and Safey El~Din}{Greuet and Safey
  El~Din}{2014}]%
        {greuet2014probabilistic}
\bibfield{author}{\bibinfo{person}{A. Greuet} {and} \bibinfo{person}{M. Safey
  El~Din}.} \bibinfo{year}{2014}\natexlab{}.
\newblock \showarticletitle{Probabilistic algorithm for polynomial optimization
  over a real algebraic set}.
\newblock \bibinfo{journal}{\emph{SIAM Journal on Optimization}}
  \bibinfo{volume}{24}, \bibinfo{number}{3} (\bibinfo{year}{2014}),
  \bibinfo{pages}{1313--1343}.
\newblock


\bibitem[\protect\citeauthoryear{Guo, Safey El~Din, and Zhi}{Guo
  et~al\mbox{.}}{2010}]%
        {guo2010global}
\bibfield{author}{\bibinfo{person}{F. Guo}, \bibinfo{person}{M. Safey El~Din},
  {and} \bibinfo{person}{L. Zhi}.} \bibinfo{year}{2010}\natexlab{}.
\newblock \showarticletitle{Global optimization of polynomials using
  generalized critical values and sums of squares}. In
  \bibinfo{booktitle}{\emph{Proceedings of the 2010 International Symposium on
  Symbolic and Algebraic Computation}}. \bibinfo{publisher}{ACM},
  \bibinfo{address}{M{\"u}nchen, Germany}, \bibinfo{pages}{107--114}.
\newblock


\bibitem[\protect\citeauthoryear{Hauenstein}{Hauenstein}{2013}]%
        {hauenstein2013numerically}
\bibfield{author}{\bibinfo{person}{J.~D. Hauenstein}.}
  \bibinfo{year}{2013}\natexlab{}.
\newblock \showarticletitle{Numerically computing real points on algebraic
  sets}.
\newblock \bibinfo{journal}{\emph{Acta applicandae mathematicae}}
  \bibinfo{volume}{125}, \bibinfo{number}{1} (\bibinfo{year}{2013}),
  \bibinfo{pages}{105--119}.
\newblock


\bibitem[\protect\citeauthoryear{Hauenstein, Safey El~Din, Schost, and
  Vu}{Hauenstein et~al\mbox{.}}{2021}]%
        {hauenstein2021solving}
\bibfield{author}{\bibinfo{person}{J.~D. Hauenstein}, \bibinfo{person}{M. Safey
  El~Din}, \bibinfo{person}{{\'E}. Schost}, {and} \bibinfo{person}{T.~X. Vu}.}
  \bibinfo{year}{2021}\natexlab{}.
\newblock \showarticletitle{Solving determinantal systems using homotopy
  techniques}.
\newblock \bibinfo{journal}{\emph{Journal of Symbolic Computation}}
  \bibinfo{volume}{104} (\bibinfo{year}{2021}), \bibinfo{pages}{754--804}.
\newblock


\bibitem[\protect\citeauthoryear{Kronecker}{Kronecker}{1882}]%
        {kronecker1882grundzuge}
\bibfield{author}{\bibinfo{person}{L. Kronecker}.}
  \bibinfo{year}{1882}\natexlab{}.
\newblock \showarticletitle{Grundz{\"u}ge einer arithmetischen Theorie der
  algebraische Gr{\"o}ssen.}
\newblock  (\bibinfo{year}{1882}), \bibinfo{pages}{1--122}.
\newblock


\bibitem[\protect\citeauthoryear{Labahn, Safey El~Din, Schost, and Vu}{Labahn
  et~al\mbox{.}}{2021}]%
        {labahn2021homotopy}
\bibfield{author}{\bibinfo{person}{G. Labahn}, \bibinfo{person}{M. Safey
  El~Din}, \bibinfo{person}{{\'E}. Schost}, {and} \bibinfo{person}{T.~X. Vu}.}
  \bibinfo{year}{2021}\natexlab{}.
\newblock \showarticletitle{Homotopy techniques for solving sparse column
  support determinantal polynomial systems}.
\newblock \bibinfo{journal}{\emph{Journal of Complexity}}  \bibinfo{volume}{66}
  (\bibinfo{year}{2021}), \bibinfo{pages}{101557}.
\newblock


\bibitem[\protect\citeauthoryear{Macaulay}{Macaulay}{1994}]%
        {macaulay1994algebraic}
\bibfield{author}{\bibinfo{person}{F.~S. Macaulay}.}
  \bibinfo{year}{1994}\natexlab{}.
\newblock \bibinfo{booktitle}{\emph{The algebraic theory of modular systems}}.
  Vol.~\bibinfo{volume}{19}.
\newblock \bibinfo{publisher}{Cambridge University Press},
  \bibinfo{address}{New York, NY, USA}.
\newblock


\bibitem[\protect\citeauthoryear{Nie, Demmel, and Sturmfels}{Nie
  et~al\mbox{.}}{2006}]%
        {nie2006minimizing}
\bibfield{author}{\bibinfo{person}{J. Nie}, \bibinfo{person}{J. Demmel}, {and}
  \bibinfo{person}{B. Sturmfels}.} \bibinfo{year}{2006}\natexlab{}.
\newblock \showarticletitle{Minimizing polynomials via sum of squares over the
  gradient ideal}.
\newblock \bibinfo{journal}{\emph{Mathematical programming}}
  \bibinfo{volume}{106}, \bibinfo{number}{3} (\bibinfo{year}{2006}),
  \bibinfo{pages}{587--606}.
\newblock


\bibitem[\protect\citeauthoryear{Rouillier}{Rouillier}{1999}]%
        {rouillier1999solving}
\bibfield{author}{\bibinfo{person}{F. Rouillier}.}
  \bibinfo{year}{1999}\natexlab{}.
\newblock \showarticletitle{Solving zero-dimensional systems through the
  rational univariate representation}.
\newblock \bibinfo{journal}{\emph{Applicable Algebra in Engineering,
  Communication and Computing}} \bibinfo{volume}{9}, \bibinfo{number}{5}
  (\bibinfo{year}{1999}), \bibinfo{pages}{433--461}.
\newblock


\bibitem[\protect\citeauthoryear{Safey El~Din and Schost}{Safey El~Din and
  Schost}{2003}]%
        {safey2003polar}
\bibfield{author}{\bibinfo{person}{M. Safey El~Din} {and} \bibinfo{person}{\'E.
  Schost}.} \bibinfo{year}{2003}\natexlab{}.
\newblock \showarticletitle{Polar varieties and computation of one point in
  each connected component of a smooth real algebraic set}. In
  \bibinfo{booktitle}{\emph{Proceedings of the 2003 international symposium on
  Symbolic and algebraic computation}}. \bibinfo{publisher}{ACM},
  \bibinfo{address}{Philadelphia, Pennsylvania, USA},
  \bibinfo{pages}{224--231}.
\newblock


\bibitem[\protect\citeauthoryear{Safey El~Din and Schost}{Safey El~Din and
  Schost}{2011}]%
        {el2011baby}
\bibfield{author}{\bibinfo{person}{M. Safey El~Din} {and}
  \bibinfo{person}{{\'E}. Schost}.} \bibinfo{year}{2011}\natexlab{}.
\newblock \showarticletitle{A baby steps/giant steps probabilistic algorithm
  for computing roadmaps in smooth bounded real hypersurface}.
\newblock \bibinfo{journal}{\emph{Discrete \& Computational Geometry}}
  \bibinfo{volume}{45}, \bibinfo{number}{1} (\bibinfo{year}{2011}),
  \bibinfo{pages}{181--220}.
\newblock


\bibitem[\protect\citeauthoryear{Safey El~Din and Schost}{Safey El~Din and
  Schost}{2017}]%
        {din2017nearly}
\bibfield{author}{\bibinfo{person}{M. Safey El~Din} {and}
  \bibinfo{person}{{\'E}. Schost}.} \bibinfo{year}{2017}\natexlab{}.
\newblock \showarticletitle{A nearly optimal algorithm for deciding
  connectivity queries in smooth and bounded real algebraic sets}.
\newblock \bibinfo{journal}{\emph{Journal of the ACM (JACM)}}
  \bibinfo{volume}{63}, \bibinfo{number}{6} (\bibinfo{year}{2017}),
  \bibinfo{pages}{1--37}.
\newblock


\bibitem[\protect\citeauthoryear{Schost}{Schost}{2003}]%
        {schost2003computing}
\bibfield{author}{\bibinfo{person}{{\'E}. Schost}.}
  \bibinfo{year}{2003}\natexlab{}.
\newblock \showarticletitle{Computing parametric geometric resolutions}.
\newblock \bibinfo{journal}{\emph{Applicable Algebra in Engineering,
  Communication and Computing}} \bibinfo{volume}{13}, \bibinfo{number}{5}
  (\bibinfo{year}{2003}), \bibinfo{pages}{349--393}.
\newblock


\bibitem[\protect\citeauthoryear{Spaenlehauer}{Spaenlehauer}{2014}]%
        {spaenlehauer2014complexity}
\bibfield{author}{\bibinfo{person}{P-J. Spaenlehauer}.}
  \bibinfo{year}{2014}\natexlab{}.
\newblock \showarticletitle{On the Complexity of Computing Critical Points with
  Gr{\"o}bner Bases}.
\newblock \bibinfo{journal}{\emph{SIAM Journal on Optimization}}
  \bibinfo{volume}{24}, \bibinfo{number}{3} (\bibinfo{year}{2014}),
  \bibinfo{pages}{1382--1401}.
\newblock


\end{thebibliography}

\end{document}